\documentclass[journal]{IEEEtran}

\usepackage[T1]{fontenc}
\usepackage{ucs}
\usepackage[utf8x]{inputenc}
  
\usepackage[greek,english]{babel}
\usepackage[babel]{csquotes}

\usepackage[noadjust]{cite}

\usepackage{fixmath}
\usepackage{mathtools}
\usepackage{amsmath}
\usepackage{amsthm}
\usepackage{amssymb}
\usepackage{graphicx}
\usepackage{tikz}
\usepackage{pgfplots}
\usepackage{algorithm}
\usepackage{algorithmic}
\usepackage{stfloats}
\usepackage{environ}
\usepackage{subeqnarray}

\usepackage{accents}
\newcommand{\ubar}[1]{\underaccent{\bar}{#1}}

\ifCLASSOPTIONdraftcls
\NewEnviron{mueq}{\begin{equation}\BODY\end{equation}}
\else
\NewEnviron{mueq}{\begin{multline}\BODY\end{multline}}
\fi

\DeclareMathAlphabet{\mathbit}{OML}{cmr}{bx}{it}
\DeclareMathAlphabet{\mathss}{T1}{lmss}{m}{it}
\DeclareMathAlphabet{\mathssbold}{T1}{lmss}{bx}{sl}
\DeclareMathAlphabet{\mathssgreek}{LGR}{lmss}{m}{sl}
\DeclareMathAlphabet{\mathssgreekbold}{LGR}{lmss}{bx}{sl}

\DeclareSymbolFont{sssymbols}{T1}{lmss}{m}{it}
\SetSymbolFont{sssymbols}{bold}{T1}{lmss}{bx}{it}
\DeclareMathAccent{\rtilde}{\mathalpha}{sssymbols}{"03}
\DeclareMathAccent{\rbar}{\mathalpha}{sssymbols}{"09}

\newcommand{\mbc}[1]{\mathssbold{#1}}

\newcommand{\mc}[1]{\mathss{#1}}
\newcommand{\mbcg}[1]{\mathssgreekbold{#1}}

\newcommand{\mb}[1]{\mathbit{#1}}
\newcommand{\mbbar}[1]{\bar{\mb{#1}}}
\newcommand{\mt}[1]{\mathrm{#1}}  

\DeclareMathOperator{\HermitianOp}{H}
\DeclareMathOperator{\TransposedOp}{T}

\newcommand{\bsc}[1]{\grave{\mb #1}}	
\newcommand{\crr}[1]{{\check{\mb #1}}}	

\newcommand{\zero}{\boldsymbol{0}} 
\newcommand{\id}{\mathbf{I}} 

\newcommand{\J}{\mathrm{j}}
\newcommand{\E}{\mathrm{e}}

\newcommand{\tr}[1]{\operatorname{trace}\lbrack#1\rbrack}	
\DeclareMathOperator{\st}{s.t.}
\DeclareMathOperator{\diag}{diag}	

\newcommand{\kron}{\otimes}

\newcommand{\Expect}[1]{\operatorname{E}\lbrack#1\rbrack}

\newcommand{\He}{{\HermitianOp}}
\newcommand{\Tr}{{\TransposedOp}}

\DeclareMathOperator*{\argmax}{argmax}

\DeclareMathOperator*{\find}{find}

\newcommand{\proj}[1]{\operatorname{proj}_{#1}}

\newcommand{\covs}[1]{c_{#1}}
\newcommand{\pcovs}[1]{\mc{\rtilde c}_{#1}}
\newcommand{\cov}[1]{\mbc C_{#1}}
\newcommand{\covR}[1]{\mb C_{#1}}
\newcommand{\pcov}[1]{\mbc{\rtilde C}_{#1}}

\newcommand{\allk}{{\forall k}}
\newcommand{\Ms}{L}
\newcommand{\ms}{\ell}
\newcommand{\Lag}{\Theta}

\newcommand{\brbfun}{F}

\newcommand{\brbinitfun}{\hat f}
\newcommand{\brbeps}{\epsilon}
\newcommand{\brbBset}{\mathbb{B}}
\newcommand{\brbB}{\mathcal{B}}
\newcommand{\brbBox}[2]{\left[{#1};~{#2}\right]}
\newcommand{\brbBoxSmall}[2]{[{#1};~{#2}]}
\newcommand{\brbUsymb}{U}
\newcommand{\brbLsymb}{A}
\newcommand{\brbU}[1][\brbBox{\brba}{\brbb}]{\brbUsymb(#1)}
\newcommand{\brbL}[1][\brbBox{\brba}{\brbb}]{\brbLsymb(#1)}
\newcommand{\brbxsymb}{x}
\newcommand{\brbysymb}{y}
\newcommand{\brbasymb}{a}
\newcommand{\brbbsymb}{b}
\newcommand{\brbpsymb}{p}
\newcommand{\brbx}{\mb\brbxsymb}
\newcommand{\brbxk}[1][k]{\brbxsymb_{#1}}
\newcommand{\brby}{\mb\brbysymb}
\newcommand{\brbyk}[1][k]{\brbysymb_{#1}}
\newcommand{\brba}{\mb\brbasymb}

\newcommand{\hatbrbak}[1][k]{{\hat\brbasymb}_{#1}}
\newcommand{\brbb}{\mb\brbbsymb}

\newcommand{\hatbrbbk}[1][k]{{\hat\brbbsymb}_{#1}}
\newcommand{\brbp}{\mb\brbpsymb}
\newcommand{\brbpk}[1][k]{\brbpsymb_{#1}}

\newcommand{\gradeps}{\epsilon}
\newcommand{\graditer}{m}
\newcommand{\gradss}{s}

\newcommand{\bnoise}{\mbcg h}

\newcommand{\pow}{p}

\newcommand{\notk}{j}

\theoremstyle{remark}
\newtheorem{theorem}{Theorem}
\newtheorem{lemma}{Lemma}
\newtheorem{remark}{Remark}

\begin{document}

\title{Two-User SIMO Interference Channel with TIN: Improper Signaling versus Time-Sharing}

\author{Christoph~Hellings,~\IEEEmembership{Member,~IEEE}, 
        Ferhad Askerbeyli,
        and~Wolfgang~Utschick,~\IEEEmembership{Senior~Member,~IEEE}
\thanks{The authors conducted this research at Technische Universit\"at M\"unchen, Professur f\"ur Methoden der Signalverarbeitung, 80290 M\"unchen, Germany, 
Telephone: +49 89 289-28520, e-mail: hellings@tum.de, ferhad.askerbeyli@tum.de, utschick@tum.de.
C. Hellings is now with Department of Physics, ETH Zurich, 8093 Zurich, Switzerland. F. Askerbeyli is also with Huawei Technologies GmbH, 80992 M\"unchen, Germany.}}

\maketitle

\begin{abstract}
In the two-user Gaussian interference channel with Gaussian inputs and treating interference as noise (TIN),
improper complex signals can be beneficial if time-sharing is not allowed or if only the data rates are averaged over several transmit strategies (convex hull formulation).
On the other hand, proper (circularly symmetric) signals have recently been shown to be optimal if coded time-sharing is considered,
i.e., if both the data rates and the transmit powers are averaged.
In this paper, we show that both conclusions remain the same if single-input multiple-output (SIMO)
systems with multiple antennas at the receivers are considered.
The proof for the case with coded time-sharing is via a novel enhanced channel concept for the two-user SIMO interference channel,
which turns out to deliver a tight outer bound to the TIN rate region with coded time-sharing.
The result for the case without coded time-sharing is demonstrated by studying specific examples
in which a newly proposed composite real gradient-projection method for improper signaling can outperform the globally optimal proper signaling strategy.
In addition, we discuss how the achievable TIN rate region with coded time-sharing can be computed numerically.
\end{abstract}


\section{Introduction}
\label{sec:intro}
Using so-called improper signals, where the pseudovariance $\mc{\rtilde c}_{\mc x}=\Expect{(\mc x-\Expect{\mc x})^2}$ is not equal to zero,
instead of proper signals \cite{NeMa93} with zero pseudovariance\footnote{In the case of zero-mean Gaussian random variables, propriety is equivalent to circular symmetry of the probability density function.} 
has been identified as a candidate to manage interference in future communication systems.
In particular, it was shown in \cite{CaJaWa10}
that improper signals can achieve more degrees of freedom than proper signals in the three-user Gaussian interference channel.
This inspired further research on improper signals in other communication scenarios with interference.

For the two-user Gaussian interference channel
under the assumption that we use Gaussian codebooks and treat interference as noise (TIN),
it has recently been shown that we have to distinguish between two cases.
It was shown in \cite{HeUt19b} that proper signaling achieves the whole rate region if coded time-sharing (see \cite{HaKo81,GaKi11} and Section~\ref{sec:sysmod:ts}) is allowed,
i.e., if it is allowed to average the data rates and the transmit powers over several transmit strategies.
If we instead restrict ourselves to the so-called convex hull formulation (see \cite{HaKo81} and Section~\ref{sec:sysmod:cvx}), i.e., we only allow an averaging of the rates,
or to pure strategies (see Section~\ref{sec:sysmod:pure}), i.e., we do not allow any averaging,
improper signaling can bring gains over proper signaling as demonstrated in \cite{ZeYeGuGuZh13}.

In this paper, we show that the situation remains the same in the two-user Gaussian single-input multiple-output (SIMO) interference channel
with Gaussian codebooks and TIN.
For the case with coded time-sharing, we provide a nontrivial extension of the proof from \cite{HeUt19a},
showing that proper signals are the optimal inputs in this case (Section~\ref{sec:main}).
Afterwards, we discuss numerical algorithms for pure strategies and for coded time-sharing under the assumption of proper signals,
and we propose a gradient-projection method for weighted sum rate maximization with improper signals (Section~\ref{sec:algo}).
Using the algorithmic solutions, we can visualize comparisons of the rate regions obtained with the various types of strategies
(proper/improper and pure/convex hull/time-sharing),
and we can establish the result that improper signals can be beneficial if coded-time sharing is not allowed (Section~\ref{sec:num}).
Finally, interpretations of the results and possible extensions to other scenarios are discussed (Section~\ref{sec:conclusion}).

There is a large variety of results on the comparison of proper and improper signals in single-input single-output (SISO) scenarios,
and some results are available for the multiple-input single-output (MISO) interference channel,
but the topic has not yet been studied in detail for the SIMO interference channel.
Improper signals were considered in the SISO interference channel from a game-theoretic perspective in \cite{HoJo12},
and heuristic optimization methods for improper signaling were proposed in \cite{PaPaKiLe13,ZeYeGuGuZh13}.
Demonstrations of the superiority of improper signaling in the case without coded time-sharing were given in \cite{ZeYeGuGuZh13,KiYeCh13}
for the case of mutual interference and in \cite{KuSu15,LaSaSc17} for the case with one-sided interference (so-called Z-interference channel).
The optimality of proper signals in the case with coded time-sharing was first shown for the Z-interference channel in \cite{HeUt17a}
and then extended to the two-user interference channel with mutual interference in \cite{HeUt19b}.
For the two-user MISO interference channel, \cite{ZeZhGuGu13} proposed a heuristic with improper signals that
was shown to outperform the globally optimal proper signaling in the case without coded time-sharing.
A heuristic improper signaling scheme for the multiple-input multiple-output (MIMO) interference channel can be found in \cite{LaAgVi16}.

All these previous results do not answer the two questions that are settled in this paper,
namely whether proper signals remain optimal in the case with coded time-sharing when switching from a SISO scenario to a SIMO scenario,
and whether gains by improper signals can be shown in the SIMO scenario without coded time-sharing.

As an additional contribution, we show that symbol extensions (considering multiple subsequent channel uses as a single channel use in a higher-dimensional system, e.g., \cite{JaSh08,CaJa08,CaJaWa10})
are not required to achieve the whole time-sharing rate region.
This is a generalization of results for the real-valued SISO interference channel in \cite{BeLiNaYa16} and the complex SISO interference channel in \cite{HeUt19b}.

Finally, the proposed gradient-projection algorithm might be of interest in its own
since it can also be applied to the more general MIMO interference channel.
A comparison of this approach to other heuristic approaches should be investigated in the future,
but goes beyond the scope of this paper,
which focuses on investigating the above-mentioned fundamental aspects of the two-user SIMO interference channel.

\emph{Notation:} 
We use $\zero$ for the zero vector or zero matrix, $\bullet^\Tr$ for the transpose, and $\bullet^\He$ for the conjugate-transpose.
Inequalities for vectors have to be understood as sets of component-wise inequalities
while $\succeq$ for matrices is meant in the sense of positive-semidefiniteness.
The vector $\mb e_i$ is the $i$th canonical unit vector, 
the matrix $\id_N$ is the $N\times N$ identity matrix,
and we use $\diag(\bullet_i)$ to construct a diagonal matrix from its diagonal elements $\bullet_i$.
We write $\kron$ for the Kronecker product, and $\|\bullet\|$ is the $2$-norm of a vector.
The ceiling operation $\lceil a\rceil$ rounds a real number $a$ to the smallest integer greater than or equal to $a$.
We use $\Re$, $\Im$, and $\angle$ for the real part, imaginary part, and the argument of a complex number, respectively.
To distinguish real-valued and complex quantities, we write complex quantities in sans-serif font and real-valued quantities in serif font.
Throughout the paper, we use the shorthand notation $\notk:=3-k$ for the index of the interfering user when considering a user $k\in\{1,2\}$.

\section{System Model and Time-Sharing}
\label{sec:sysmod}
We consider a two-user SIMO interference channel 
\begin{subequations}
\label{eq:model}
\begin{align}
\label{eq:model1}
\mbc y_1 &= \mbc h_{11} \mc x_1 + \mbc h_{12} \mc x_2 + \bnoise_1 \in\mathbb{C}^{N_1} \\
\label{eq:model2}
\mbc y_2 &= \mbc h_{21} \mc x_1 + \mbc h_{22} \mc x_2 + \bnoise_2 \in\mathbb{C}^{N_2}
\end{align}
\end{subequations}
with proper Gaussian noise $\bnoise_k\sim\mathcal{CN}(\zero,\cov{\bnoise_k})$, where
the input signals $\mc x_k,~k=1,2$ are (possibly improper) zero-mean complex Gaussian with variance $\covs{\mc x_k}=\Expect{|\mc x_k|^2}$ and pseudovariance $\pcovs{\mc x_k}=\Expect{\mc x_k^2}$.
The whole paper focuses on the case where the receivers treat interference as noise (TIN).
It is thus convenient for some derivations to define the interference-plus-noise signals
\begin{align}
\mbc s_1 &= \mbc h_{12} \mc x_2 + \bnoise_1, &
\mbc s_2 &= \mbc h_{21} \mc x_1 + \bnoise_2.
\end{align}

The channel vectors $\mbc h_{kk}$ and $\mbc h_{k\notk}$ are assumed to stay constant over time, and we assume perfect channel state information.
Furthermore, it is assumed that the input signals of both users and the noise at both users (i.e., $\mc x_1$, $\mc x_2$, $\bnoise_1$, and $\bnoise_2$) are all mutually independent.
Without loss of generality, we assume $\cov{\bnoise_k}=\id_{N_k},~\allk$.

The achievable rates (Shannon rates) with TIN can be written as (e.g., \cite{ZeYeGuGuZh13})
\begin{multline}
\label{eq:rk}
r_k(\mathcal X) =\\ \log_2\frac{\det \cov{\mbc y_k} }{ \det \cov{\mbc s_k}  }+\frac{1}{2}\log_2
\frac{\det\left(\id_{N_k} - \cov{\mbc y_k}^{-1} \pcov{\mbc y_k} \cov{\mbc y_k}^{-\Tr} \pcov{\mbc y_k}^\He\right)}
{\det\left(\id_{N_k} - \cov{\mbc s_k}^{-1} \pcov{\mbc s_k} \cov{\mbc s_k}^{-\Tr} \pcov{\mbc s_k}^\He\right)}
\end{multline}
with 
\begin{subequations}
\begin{align}
\cov{\mbc y_k} &= \mbc h_{kk} \covs{\mbc x_k} \mbc h_{kk}^\He + \cov{\mbc s_k}, &
\cov{\mbc s_k} &= \mbc h_{k\notk} \covs{\mbc x_\notk} \mbc h_{k\notk}^\He + \id_{N_k}, \\
\pcov{\mbc y_k} &= \mbc h_{kk} \pcovs{\mbc x_k} \mbc h_{kk}^\Tr + \pcov{\mbc s_k}, &
\pcov{\mbc s_k} &= \mbc h_{k\notk} \pcovs{\mbc x_\notk} \mbc h_{k\notk}^\Tr.
\end{align}
\end{subequations}
The tuple $\mathcal X=(c_{\mc x_1},c_{\mc x_2},\mc{\rtilde c}_{\mc x_1},\mc{\rtilde c}_{\mc x_2})$ 
summarizes all parameters that describe the chosen strategy, i.e., all transmit variances and pseudovariances.
The special case $\mc{\rtilde c}_{\mc x_1}=\mc{\rtilde c}_{\mc x_2}=0$ corresponds to proper signaling.
In this case, the second summand in \eqref{eq:rk} vanishes.

As an alternative to \eqref{eq:rk}, we can calculate the achievable data rates via the composite real representation,
where complex vectors $\mbc b$ and linear operations $\mbc b \mapsto \mbc A \mbc b$ (with a complex matrix $\mbc A$) are represented by (e.g., \cite{AdScSc11})
\begin{align}
\crr b&= \begin{bmatrix}
\Re\mbc b\\\Im\mbc b
\end{bmatrix}
&&\text{and}&
\label{eq:crrmat}
\crr b &\mapsto \bsc A \crr b , ~~\bsc A = \begin{bmatrix}
\Re\mbc A & -\Im\mbc A\\\Im\mbc A& \Re\mbc A
\end{bmatrix}.
\end{align}
The second-order properties of a random vector $\mbc b$ with covariance matrix $\mbc C_{\mbc b}$ and pseudocovariance matrix $\mbc{\rtilde C}_{\mbc b}$,
can equivalently be described by the covariance matrix of $\crr b$.
The relation reads as (e.g., \cite{HeUt15})
\begin{align}
\mb C_{\crr b} &= \frac{1}{2}\left(\begin{bmatrix}
\Re \mbc{C}_{\mbc b} & - \Im \mbc{C}_{\mbc b} \\
\Im \mbc{C}_{\mbc b} & \Re \mbc{C}_{\mbc b}
\end{bmatrix}+\begin{bmatrix}
\Re \mbc{\rtilde C}_{\mbc b} & \Im \mbc{\rtilde C}_{\mbc b} \\
\Im \mbc{\rtilde C}_{\mbc b} & - \Re \mbc{\rtilde C}_{\mbc b}
\end{bmatrix}\right)
\label{eq:crrcov}
\end{align}
and $\mb C_{\crr b}$ is referred to as the composite real covariance matrix.

Using these definitions, the achievable rates are given by
\begin{equation}
\label{eq:rkR}
r_k(\covR{\crr x_k},\covR{\crr x_\notk}) =\\ \frac{1}{2}\log_2\frac{\det \covR{\crr y_k} }{ \det \covR{\crr s_k}  }
\end{equation}
with
\begin{subequations}
\begin{align}
\covR{\crr y_k} &= \bsc H_{kk} \covR{\crr x_k} \bsc H_{kk}^\Tr + \covR{\crr s_k}, \\
\covR{\crr s_k} &= \bsc H_{k\notk} \covR{\crr x_\notk} \bsc H_{k\notk}^\Tr + \frac{1}{2}\id_{2N_k}.
\label{eq:covRsk}
\end{align}
\end{subequations}
The composite real channel matrices according to the definition in \eqref{eq:crrmat} read as
\begin{align}
\bsc H_{kk} &= \begin{bmatrix}
\Re\mbc h_{kk} & -\Im\mbc h_{kk}\\\Im\mbc h_{kk}& \Re\mbc h_{kk}
\end{bmatrix}, &
\bsc H_{k\notk} &= \begin{bmatrix}
\Re\mbc h_{k\notk} & -\Im\mbc h_{k\notk}\\\Im\mbc h_{k\notk}& \Re\mbc h_{k\notk}
\end{bmatrix}
\end{align}
The factor of $\frac{1}{2}$ in front of the logarithm accounts for the fact that real-valued data streams instead of complex ones are considered
while the factor of $\frac{1}{2}$ in \eqref{eq:covRsk} comes from applying \eqref{eq:crrcov} to the noise covariance matrix.

\subsection{Pure Strategies}
\label{sec:sysmod:pure}
In this paper, three types of transmit strategies are considered.
In the first type, which we refer to as \emph{pure} strategies, a single choice for the statistical properties of the input signals is applied as long as the channel realization remains unchanged.

To study the rate region with pure strategies we compute Pareto-optimal pairs of achievable rates $(r_1,r_2)$ by
solving the so-called \emph{rate balancing} \cite{JoBo02} optimization\footnote{Note that \eqref{eq:noTS:pos} is the necessary  and sufficient condition for a valid pseudovariance (see, e.g., \cite{ScSc10}).}
\begin{subequations}
\label{eq:noTS}
\begin{align}
\max_{\mathcal X,R\in\mathbb{R}} ~~R  \quad\st\quad & r_k(\mathcal X) \geq \rho_k R,~\allk \label{eq:noTS:rate} \\
& 0\leq c_{\mc x_k}\leq P_k, ~\allk \label{eq:noTS:pow} \\
&  |\mc{\rtilde c}_{\mc x_k}| \leq c_{\mc x_k},~\allk \label{eq:noTS:pos}
\end{align}
\end{subequations}
with the so-called \emph{rate profile vector} \cite{MoZhCi06} $\mb\rho$ set to $\mb\rho=[\rho_1,\rho_2]^\Tr=[\beta,1-\beta]^\Tr$ for various $\beta\in[0;1]$.
The entries of $\rho_k$ define relative rate targets of the two users,
and the optimal value of $R$ is the highest possible common scaling factor that still leads to a feasible pair of rates.
Without loss of generality, we assume that $\rho_1+\rho_2=1$, so that the value of $R$ equals the achieved sum rate.

\subsection{Time-Sharing}
\label{sec:sysmod:ts}
The second type of strategies we consider combine multiple transmit strategies by means of weighting factors $\mb{\tau}=[\tau_1, \dots, \tau_\Ms]$ that indicate
which fraction of the total time the $\ms$th strategy should be used.
Instead of interpreting them as the length of time intervals,
these weights can also be seen as the probability that the $\ms$th strategy is chosen
(see the concept of a time-sharing parameter in \cite{HaKo81}).

The so-called \emph{time-sharing} strategies (or \emph{coded time-sharing}, e.g., \cite{GaKi11})
can be optimized by solving
\begin{subequations}
\label{eq:primal}
\begin{align}
\label{eq:primal:rate}
\max_{\substack{\mathcal X^{(\ms)},\Ms\in\mathbb{N},R\in\mathbb{R}\\\tau\geq\zero:\, \sum_{\ms=1}^{\Ms}\tau_\ms=1}} ~R  \quad\st\quad 
& \sum_{\ms=1}^\Ms \tau_\ms r_k(\mathcal X^{(\ms)}) \geq \rho_k R,~\allk \\
& \sum_{\ms=1}^\Ms \tau_\ms c_{\mc x_k}^{(\ms)}\leq P_k, ~\allk \label{eq:primal:pow}\\
& 0\leq c_{\mc x_k}^{(\ms)}, ~\allk,~\forall\ms \label{eq:primal:pos}\\
&  |\mc{\rtilde c}_{\mc x_k}^{(\ms)}| \leq c_{\mc x_k}^{(\ms)},~\allk, ~\forall \ms 
\label{eq:primal:pcov}
\end{align}
\end{subequations}
where we use $\mathcal X^{(\ms)}=(c_{\mc x_1}^{(\ms)},c_{\mc x_2}^{(\ms)},\mc{\rtilde c}_{\mc x_1}^{(\ms)},\mc{\rtilde c}_{\mc x_2}^{(\ms)})$
to denote the transmit (pseudo)variances of the $\ms$th strategy.

\subsection{Convex Hull}
\label{sec:sysmod:cvx}
The third type of strategies, which do not exploit the full potential of time-sharing \cite{HaKo81},
are obtained by first finding the rate region with pure strategies and then taking its convex hull (e.g., \cite{HoJo12,ZeYeGuGuZh13,ZeZhGuGu13,LaSaSc17}).
This convex hull operation corresponds to averaging the data rates as in \eqref{eq:primal:rate}, but the power constraints for pure strategies \eqref{eq:noTS:pow}
are respected when computing the original rate region.
Thus, the convex hull formulation cannot exploit the potential of average power constraints as in \eqref{eq:primal:pow}.
This can be formulated as
\begin{subequations}
\label{eq:rts}
\begin{align}
\label{eq:rts:rate}
\max_{\substack{\mathcal X^{(\ms)},\Ms\in\mathbb{N},R\in\mathbb{R}\\\tau\geq\zero:\, \sum_{\ms=1}^{\Ms}\tau_\ms=1}} ~R  \quad\st\quad 
& \sum_{\ms=1}^\Ms \tau_\ms r_k(\mathcal X^{(\ms)}) \geq \rho_k R,~\allk \\
& 0\leq c_{\mc x_k}^{(\ms)}\leq P_k, ~\allk,~\forall\ms \label{eq:rts:pow}\\
&  |\mc{\rtilde c}_{\mc x_k}^{(\ms)}| \leq c_{\mc x_k}^{(\ms)},~\allk, ~\forall \ms.
\label{eq:rts_pcov}
\end{align}
\end{subequations}

Due to the more restrictive power constraints, the convex hull formulation leads in general to a smaller rate region than coded time-sharing.
To obtain the complete time-sharing rate region, we thus cannot take the convex hull after solving \eqref{eq:noTS},
but we instead have to account for the possibility of time-sharing already during the optimization by solving \eqref{eq:primal}.
Some might argue that coded time-sharing can lead to stronger fluctuations of the transmit powers than the convex hull formulation.
A detailed discussion why this should not be seen as an obstacle and how long-term fluctuations could be circumvented is given in \cite{HeUt19b}.

\subsection{Remark on Symbol Extensions}
\label{sec:sysmod:ext}
All rate expressions given above can be extended to include the possibility of symbol extensions (see, e.g., \cite{JaSh08,CaJa08,CaJaWa10,BeLiNaYa16,NaNg19}).
In this case, transmit symbols
\begin{equation}
\label{eq:symext}
\ubar {\mbc x}_k
=\begin{bmatrix}
\mc x_{k,1}\\\vdots\\\mc x_{k,T}
\end{bmatrix}
\end{equation}
spanning over $T$ channel uses are considered,
and the achievable rates can be calculated via
\begin{equation}
\label{eq:rkRT}
r_k(\mathcal X) =\\ \frac{1}{2T}\log_2\frac{\det \covR{\ubar{\crr y}_k} }{ \det \covR{\ubar{\crr s}_k}  }
\end{equation}
with
\begin{subequations}
\begin{align}
\covR{\ubar{\crr y}_k} &= \ubar{\bsc H}_{kk} \covR{\ubar{\crr x}_k} \ubar{\bsc H}_{kk}^\Tr + \covR{\ubar{\crr s}_k}, \\
\covR{\ubar{\crr s}_k} &= \ubar{\bsc H}_{k\notk} \covR{\ubar{\crr x}_\notk} \ubar{\bsc H}_{k\notk}^\Tr + \frac{1}{2}\id_{2TN_k}.
\label{eq:covRskExt}
\end{align}
\end{subequations}
where $\ubar{\bsc H}_{kk}$ and $\ubar{\bsc H}_{k\notk}$ are the composite real representations \eqref{eq:crrmat} of\footnote{Recall that we only consider constant channel coefficients in this paper.}
\begin{align}
\ubar{\mbc H}_{kk}&=\id_T \kron \mbc h_{kk},&
\ubar{\mbc H}_{k\notk}&=\id_T \kron \mbc h_{k\notk}.
\end{align}

Note that switching to a symbol extension over $T$ channel uses implies that the transmit power constraint becomes an average power constraint
\begin{equation}
\frac{1}{T}\tr{\covR{\ubar{\crr x}_k}}\leq P_k
\end{equation}
over the $T$ elements of $\ubar {\mbc x}_k$.
As this does not seem to be compatible with the assumptions in the case of pure strategies or of the convex hull formulation,
we only consider the possibility of symbol extensions when studying coded time-sharing.

However, in the formal proofs in the appendix it is shown that symbol extensions do not bring any advantages if coded time-sharing is considered.
Therefore, for the sake of readability, the derivations in the following section are directly written down for the case without symbol extensions.

The fact that symbol extensions do not bring an advantage in the case of Gaussian signals and TIN with coded time-sharing
was established for the real-valued single-antenna interference channel in \cite{BeLiNaYa16}
and extended to the complex single-antenna interference channel in \cite{HeUt19b}.
Similar considerations for single-antenna scenarios can also be found in \cite{ChVe93,NaNg19}.
One part of proving Theorem~\ref{th:main} in the next section is to extend this result to the complex SIMO interference channel.

\section{Main Result with Coded Time-Sharing}
\label{sec:main}
In this section, we establish the main result of this paper, which states the optimality of proper signaling without symbol extensions
in the considered scenario with coded time-sharing.

\begin{theorem}
\label{th:main}
Consider the two-user Gaussian SIMO interference channel \eqref{eq:model}
with Gaussian input signals under power constraints \eqref{eq:primal:pow}, and assume that interference is treated as noise.
Then, the whole time-sharing rate region $\mathcal{R}$ can be achieved using proper input signals without symbol extensions.
\end{theorem}

The proof is established by combining four Lemmas that are stated and proven below.
First, Lemma~\ref{lem:trans} describes a transformation to a simpler SIMO interference channel
whose rate region $\mathcal{R}'$ equals the original rate region $\mathcal{R}$.
Then, Lemma~\ref{lem:enhanced} introduces an enhanced SIMO interference channel whose rate region $\bar{\mathcal{R}}$ contains $\mathcal{R}'$.
However, under a restriction to proper signaling without symbol extensions,
the rate regions $\mathcal{R}'_\mt{proper}$ of the transformed system and $\bar{\mathcal{R}}_\mt{proper}$ of the enhanced system coincide,
which is stated in Lemma~\ref{lem:proper_equal}.
Finally, Lemma~\ref{lem:proper_opt} studies the enhanced system and shows that proper signaling without symbol extensions achieves the whole time-sharing rate region,
i.e., $\bar{\mathcal{R}}_\mt{proper}$ is the same as $\bar{\mathcal{R}}$.

\begin{proof}[Proof of Theorem~\ref{th:main}]
By Lemmas~\ref{lem:trans}, \ref{lem:enhanced}, \ref{lem:proper_equal}, and~\ref{lem:proper_opt}, we have
$\mathcal{R}=\mathcal{R}'\subseteq\bar{\mathcal{R}}=\bar{\mathcal{R}}_\mt{proper}=\mathcal{R}_\mt{proper}'=\mathcal{R}_\mt{proper}$,
where $\mathcal{R}'$ is defined in Lemma~\ref{lem:trans}, 
$\bar{\mathcal{R}}$ is defined in Lemma~\ref{lem:enhanced}, and 
the subscript ${}_\mt{proper}$ denotes the respective rate region under a restriction to proper input signals without symbol extensions.
On the other hand, it is clear that $\mathcal{R}_\mt{proper}\subseteq\mathcal{R}$.
This shows that $\mathcal{R}_\mt{proper}=\mathcal{R}$.
\end{proof}

\begin{lemma}
\label{lem:trans}
Consider the reduced QR decomposition
\begin{equation}
\begin{bmatrix}
\mbc h_{kk} & \mbc h_{k\notk}
\end{bmatrix} = \underbrace{\mbc Q_k}_{\in\mathbb{C}^{N_k\times 2}} \begin{bmatrix}
h_k & a_k \E^{\J\varphi_k} \\
0 & b_k \E^{\J\psi_k}
\end{bmatrix}
\end{equation}
with $h_k, a_k, b_k, \varphi_k, \psi_k \in \mathbb{R}$,
and let
\begin{align}
\label{eq:modelTrans:chan}
\mb h_{11}' &= \begin{bmatrix}h_1 \\ 0 \end{bmatrix},
&
\mbc h_{22}' &= \begin{bmatrix}h_2\E^{\J\theta} \\ 0 \end{bmatrix},
&
\mb h_{12}' &= \begin{bmatrix} a_1 \\ b_1 \end{bmatrix},
&
\mb h_{21}' &= \begin{bmatrix} a_2 \\ b_2 \end{bmatrix}
\end{align}
with $\theta=-\varphi_1-\varphi_2$.
Define $\mathcal R'$ to be the time-sharing rate region of the transformed SIMO interference channel
\begin{subequations}
\label{eq:modelTrans}
\begin{align}
\label{eq:modelTrans1}
\mbc y_1 &= \mb h_{11}' \mc x_1 + \mb h_{12}' \mc x_2 + \bnoise_1' \in\mathbb{C}^{2} \\
\label{eq:modelTrans2}
\mbc y_2 &= \mb h_{21}' \mc x_1 + \mbc h_{22}' \mc x_2 + \bnoise_2' \in\mathbb{C}^{2}
\end{align}
\end{subequations}
with proper Gaussian noise $\bnoise_k'\sim\mathcal{CN}(\zero,\id_2)$.
Then, under the assumptions of Theorem~\ref{th:main},
$\mathcal{R}=\mathcal{R}'$ and any rate vector that is achievable in one system with proper signaling without symbol extensions
is also achievable in the other system under the same restrictions.
\end{lemma}
Note that there are various ways to transform multiantenna interference channels to simpler formulations, preferably of reduced dimension.
For example, a standard form of the two-user MIMO interference channel discussed in \cite[Sec.~2.2.1]{VeGa18} could also be applied to the SIMO scenario.
The transformation that we instead propose in Lemma~\ref{lem:trans} is designed in a way that the enhanced scenario needed for the following proofs can be easily created.
The proof of Lemma~\ref{lem:trans} in the appendix is based on the fact that neither a receive filtering with $\mbc Q_k^\He$ nor some required phase rotations
of the input and output signals change the achievable rates.

\begin{lemma}
\label{lem:enhanced}
Let $\bar{\mathcal R}$ denote the time-sharing rate region of the modified interference channel
\begin{subequations}
\label{eq:modelUB}
\begin{align}
\label{eq:modelUB1}
\mbc y_1 &= \mb h_{11}' \mc x_1 + \mb h_{12}' \mc x_2 + \bnoise_1' \in\mathbb{C}^{2} \\
\label{eq:modelUB2}
\mbc y_2 &= \mb h_{21}' \mc x_1 + \mbbar h_{22} \mc x_2 + \bnoise_2' \in\mathbb{C}^{2}
\end{align}
\end{subequations}
with 
\begin{equation}
\mbbar h_{kk} = \begin{bmatrix} h_k \\ 0 \end{bmatrix}
\end{equation}
and all other definitions as in Lemma~\ref{lem:trans}.
Then, under the assumptions of Theorem~\ref{th:main}, $\mathcal{R}\subseteq\bar{\mathcal R}$.
\end{lemma}

The formal proof of Lemma~\ref{lem:enhanced} including the possibility of symbol extensions is presented separately in the appendix.
In the following, we give a short justification by applying the composite real rate expression \eqref{eq:rkR} 
to the transformed system \eqref{eq:modelTrans} obtained in Lemma~\ref{lem:trans}.

Let us parameterize the composite real input covariance matrices as\footnote{The time slot index $\ms$ can be omitted for the sake of brevity whenever we consider only a particular time slot.}
\begin{equation}
\label{eq:UBparam}
\covR{\crr x_k} = 
\frac{\covs{\mc x_k}}{2}\id_2 + \frac{|\pcovs{\mc x_k}|}{2} \begin{bmatrix}
\cos\alpha_k \quad  &\sin\alpha_k \\
\sin\alpha_k  \quad &-\cos\alpha_k
\end{bmatrix}
\end{equation}
where $\alpha_k$ can be chosen arbitrarily without affecting the power constraints or the condition for a valid $\pcovs{\mc x_k}$.
The following facts can be verified by expanding the determinants by the help of a software for symbolic calculations.
First, $\det \covR{\crr s_1}$ and $\det \covR{\crr s_2}$ both do not depend on any of $\alpha_1$, $\alpha_2$ and $\theta$.
Second, there is no individual dependence of $\det \covR{\crr y_1}$ and $\det \covR{\crr y_2}$ on $\alpha_1$, $\alpha_2$, or $\theta$,
but instead the dependence is only via
\begin{align}
\label{eq:UBbeta}
\beta_1 &= \alpha_2-\alpha_1, &\beta_2 &= \alpha_2-\alpha_1+2\theta
\end{align}
respectively.
Third, we have
\begin{equation}
\label{eq:UBderiv}
\frac{\partial\det\covR{\crr y_k}}{\partial\beta_k} =
\frac{ h_k^2 |a_k|^2 |\pcovs{\mc x_1}| |\pcovs{\mc x_\notk}| \sin\beta_k}{8}.
\end{equation}

Let ${\bar r}_k(\mathcal X)$ denote a version of \eqref{eq:rkR} where, after expanding the determinants, all instances of $\beta_k$ are replaced by $\pi$
while the occurrences of $\covs{\mc x_k}$, $\covs{\mc x_\notk}$, $|\pcovs{\mc x_k}|$, and $|\pcovs{\mc x_\notk}|$ remain unchanged.
Due to \eqref{eq:UBderiv}, having $\beta_k=\pi,~\allk$ would be optimal in terms of achievable rates.
The interpretation of this is that the impropriety of the intended signal and of the received interference should point exactly in opposite directions.\footnote{Similar
observations can be found in other system models, e.g., \cite{LaSaSc17,HeUt19b}.}
However, due to \eqref{eq:UBbeta}, it is in general not possible to achieve $\beta_k=\pi$ for both users simultaneously.
Thus, ${\bar r}_k(\mathcal X)$ is in general not achievable, but it is an upper bound to $r_k(\mathcal X)$.

By contrast, the upper bounds are achievable for both users simultaneously
if we instead consider the enhanced system \eqref{eq:modelUB}, where $\theta=0$.
Moreover, the upper bounds in the enhanced system \eqref{eq:modelUB} are the same as in the transformed system \eqref{eq:modelTrans}
since the only dependence of $\det \covR{\crr y_2}$ on $\theta$ via $\beta_2$ was eliminated when forming the upper bound.
This leads to the statement of Lemma~\ref{lem:enhanced} that for any time-sharing solution with $\alpha_1$ and $\alpha_2$ being chosen optimally in each strategy,
the average rates achieved in \eqref{eq:modelUB} are at least as high as in \eqref{eq:modelTrans}.

The idea of channel enhancement was originally proposed in \cite{WeStSh06} to make a nondegraded MIMO broadcast channel degraded by increasing some channel gains
(or, equivalently, reducing the noise).
The idea was later transferred to the MIMO wiretap channel \cite{LiSh09} and the MIMO relay channel \cite{GeHeWeUt15}.
However, the channel enhancement argument that we use in this work is instead based on a different idea that was developed in \cite{HeUt19b}
for the SISO interference-channel.
Lemma~\ref{lem:enhanced} provides a nontrivial extension of this idea to the considered SISO scenario.
Compared to the classical enhancement argument from \cite{WeStSh06},
the first difference is that our aim is not to create a degraded scenario,
but rather a scenario with real-valued channel coefficients,
which turns out to be useful in the following proofs.
The second difference is that we obtain the enhancement by only changing the phase of a single channel coefficient while keeping all magnitudes unchanged.

\begin{lemma}
\label{lem:proper_equal}
Assume a constraint that all transmit signals have to be proper without symbol extensions, and 
let $\mathcal{R}'_\mt{proper}$ and $\bar{\mathcal{R}}_\mt{proper}$ denote the resulting time-sharing rate regions
of \eqref{eq:modelTrans} and \eqref{eq:modelUB}, respectively.
Then, under the assumptions of Theorem~\ref{th:main}, $\mathcal{R}'_\mt{proper}=\bar{\mathcal{R}}_\mt{proper}$.
\end{lemma}
\begin{proof}
For proper signals without time-sharing, the rate is given by the first summand of \eqref{eq:rk}.
It is easy to verify that $\cov{\mbc s_1}$, $\cov{\mbc s_2}$, and $\cov{\mbc y_1}$ do not depend on $\theta$
as the only dependence on $\theta$ can be via $\mbc h_{22}'$.
Since
\begin{equation}
\mbc h_{22}' \covs{\mbc x_k} \mbc h_{22}^{\prime\He} = \begin{bmatrix}
h_2^2 \covs{\mc x_2} &0 \\ 0 &0
\end{bmatrix}
\end{equation}
$\cov{\mbc y_2}$ does not depend on $\theta$ either. Thus, choosing $\theta=0$ to obtain \eqref{eq:modelUB} does not change the achievable rates.
\end{proof}

\begin{lemma}
\label{lem:proper_opt}
For the enhanced SIMO interference channel \eqref{eq:modelUB} under the assumptions of Theorem~\ref{th:main},
proper signaling without symbol extensions achieves the whole time-sharing rate region, i.e., 
$\bar{\mathcal{R}}_\mt{proper}=\bar{\mathcal{R}}$.
\end{lemma}

The formal proof including the consideration of symbol extensions is given in the appendix.
For the following intuitive justification, we note that all channel vectors in \eqref{eq:modelUB} are real-valued,
so that the composite real representation of each of them is a block-diagonal matrix with two equal blocks due to \eqref{eq:crrmat}.
We thus have
\begin{subequations}
\label{eq:modelBD}
\begin{align}
\crr y_1 &= \begin{bmatrix}\mb h_{11}' & 0 \\ 0 & \mb h_{11}'\end{bmatrix} \crr x_1 + \begin{bmatrix}\mb h_{12}' & 0 \\ 0 & \mb h_{12}'\end{bmatrix} \crr x_2 + \crr\eta_1\\
\crr y_2 &= \begin{bmatrix}\mb h_{21}' & 0 \\ 0 & \mb h_{21}'\end{bmatrix} \crr x_1 + \begin{bmatrix}\mbbar h_{22} & 0 \\ 0 & \mbbar h_{22}\end{bmatrix} \crr x_2 + \crr\eta_2
\end{align}
\end{subequations}
with real-valued Gaussian noise $\crr\eta_k\sim\mathcal{N}(\zero,\frac{1}{2}\id_4),~\forall k$.
This is mathematically equivalent to a symbol extension over two channel uses in a real-valued SIMO interference channel with constant channels.

We can now make use of the upper bounds on the per-user rates defined below \eqref{eq:UBderiv}.
These bounds are achievable in the enhanced scenario if we choose $\alpha_2=\alpha_1+\pi$, and they do not depend on the individual values of $\alpha_1$ and $\alpha_2$.
We may thus assume $\alpha_1=0$ without loss of generality, so that
$\mc{\rtilde c}_{\mc x_k}^{(\ms)}$ are real-valued for both users $k$.
Then, $\covR{\crr x_k}$ in \eqref{eq:UBparam} are diagonal for both $k$ and can be reparameterized as
\begin{align}
\mb C_{\crr x_k}^{(\ms)} &= \frac{1}{2}\begin{bmatrix}
c_{\mc x_k}^{(\ms)} + \mc{\rtilde c}_{\mc x_k}^{(\ms)} & 0 \\
0 & c_{\mc x_k}^{(\ms)} - \mc{\rtilde c}_{\mc x_k}^{(\ms)}
\end{bmatrix}=: \begin{bmatrix}
\pow_{k,1}^{(\ms)} & 0 \\
0 & \pow_{k,2}^{(\ms)}
\end{bmatrix}
\label{eq:per_carrier_powers}
\end{align}
in the $\ms$th time slot.
Since the determinant of a block-diagonal matrix can be rewritten as a product of determinants,
and since the logarithm of a product is a sum of logarithms, 
the diagonal covariance matrices lead to
\begin{align}
\label{eq:rk_carriers}
{\bar r}_{k}^{(\ms)} &= \sum_{t=1}^2 
\frac{1}{2}\log_2\frac{\det \covR{\crr y_{k,t}^{(\ms)}} }{ \det \covR{\crr s_{k,t}^{(\ms)}}  }
\end{align}
with
\begin{subequations}
\begin{align}
\covR{\crr y_{k,t}}^{(\ms)} &= \mbbar h_{kk} \pow_{k,t}^{(\ms)} \mbbar h_{kk}^\Tr + \covR{\crr s_{k,t}^{(\ms)}}, \\
\covR{\crr s_{k,t}^{(\ms)}} &= \mb h_{k\notk}' \pow_{\notk,t}^{(\ms)} \mb h_{k\notk}^{\prime\Tr} + \frac{1}{2}\id_{2}
\end{align}
\end{subequations}
for $t\in\{1,2\}$.

Now assume that $\pow_{k,1}^{(\ms)} \neq \pow_{k,2}^{(\ms)}$ for some $k$ and some $\ms$ in the optimal rate balancing solution.
Then, we can create a new solution with $L'=2L$ time slots with $\tau_{\ms}'=\tau_{\lceil \ms/2 \rceil} / 2$ in \eqref{eq:primal},
and set\footnote{Similar arguments have previously been used in SISO scenarios \cite{NaNg19,HeUt19b}.}
\begin{align}
\pow_{k,1}^{\prime(\ms)} = \pow_{k,2}^{\prime(\ms)} =\begin{cases}
\pow_{k,1}^{(\lceil\ms/2\rceil)},\quad \text{$\ms$ odd}, \\
\pow_{k,2}^{(\lceil\ms/2\rceil)},\quad \text{$\ms$ even}.
\end{cases}
\end{align}
This does not change the value on the left hand side of \eqref{eq:primal:pow},
and due to \eqref{eq:rk_carriers}, the value on the left hand side of \eqref{eq:primal:rate} remains unchanged as well.
In essence, since the channel model in \eqref{eq:modelBD} consists of block-diagonal matrices with equal blocks, varying the power over the blocks in a single time slot is completely interchangeable with varying the power over two consecutive time slots. Thus, any power imbalance between the blocks can be equalized by replacing it with a power imbalance over time slots.
Consequently, there always exists an optimal solution with
\begin{equation}
\pow_{k,1}^{\prime(\ms)} = \pow_{k,2}^{\prime(\ms)} ~\Leftrightarrow~
\frac{c_{\mc x_k}^{\prime(\ms)} + \mc{\rtilde c}_{\mc x_k}^{\prime(\ms)}}{2} = \frac{c_{\mc x_k}^{\prime(\ms)} - \mc{\rtilde c}_{\mc x_k}^{\prime(\ms)}}{2} ~\Leftrightarrow~
\mc{\rtilde c}_{\mc x_k}^{\prime(\ms)}=0
\end{equation}
for all users $k$ and all time slots $\ms$,
i.e., a solution with proper signaling.

To complete the proof of Lemma~\ref{lem:proper_opt} and thus of Theorem~\ref{th:main}, the argumentation is extended 
in the formal proof in the appendix in order to include the possibility of symbol extensions in the complex setting.

\section{Algorithmic Solutions}
\label{sec:algo}
In this section, we discuss numerical methods to optimize the various types of transmit strategies discussed in this paper.
Under a restriction to proper signals, we comment on a globally optimal method for rate balancing with pure strategies, and
we propose a globally optimal method for rate balancing with coded time-sharing.
Afterwards, we turn our attention to improper signals and propose a heuristic approach to weighted sum rate maximization,
which we then use to draw conclusions about pure strategies and about the convex hull formulation in Section~\ref{sec:num}.
Note that there is no need to derive an optimization method for coded time-sharing with improper signals since
we have shown in Theorem~\ref{th:main} that proper signals are optimal in case of coded time-sharing.

\subsection{Pure Strategies with Proper Signals}
\label{sec:algo:pure}
In \cite{LiZhCh12}, a weighted sum rate maximization in the $K$-user SIMO interference channel with proper signals was considered,
and an exponential-complexity method based on monotonic optimization was proposed for this nonconvex problem.
An arising subproblem in this method is the rate balancing problem \eqref{eq:noTS},
for which an efficient solution was proposed in \cite{LiZhCh12}.
For the reader's convenience, we briefly summarize the relevant steps below.

Via a bisection over $R$, problem~\eqref{eq:noTS} can be turned into a series of feasibility problems with fixed rate targets (due to fixed $R$) instead of relative rate targets, i.e.,
\begin{subequations}
\label{eq:noTSbisect}
\begin{align}
\find_{\substack{\covs{\mbc x_1},\covs{\mbc x_2}\\\mbc w_1,\mbc w_2}} ~\st~ & \log_2\left(
1 + \gamma_k\right)
 \geq \rho_k R,~\allk  \\
& 0 \leq \covs{\mbc x_k} \leq P_k, ~\allk 
\end{align}
\end{subequations}
where the signal-to-interference-and-noise ratio $\gamma_k$
can been expressed by means of a receive filter $\mbc w_k$ as
\begin{equation}
\gamma_k = \frac{|\mbc w_k^\He\mbc h_{kk}|^2 \covs{\mbc x_k}}{|\mbc w_k^\He\mbc h_{k\notk}|^2 \covs{\mbc x_\notk} +  \mbc w_k^\He \mbc w_k}.
\end{equation}
Problem \eqref{eq:noTSbisect} has a solution if and only if the value of
\begin{equation}
\label{eq:noTSbisectSINR}
\Gamma(R) = \left(\max_{\substack{\covs{\mbc x_1},\covs{\mbc x_2}\\\mbc w_1,\mbc w_2}} 
\min_{k\in\{1,2\}} 
\frac{\gamma_k}{2^{\rho_k R}-1}
~\st~
  0 \leq \covs{\mbc x_k} \leq P_k,~\allk\right)
\end{equation}
is larger than or equal to $1$.

According to \cite[Th.~4.1]{LiZhCh12}, there exists an $i\in\{1,2\}$ such that the solution of \eqref{eq:noTSbisectSINR} remains unchanged if
the power constraint of user $k=i$ is ignored.
Moreover, problem \eqref{eq:noTSbisectSINR} with only one power constraint can be solved via Perron-Frobenius theory (e.g., \cite{HoJo13})
as derived in detail in \cite{LiZhCh12}.

These ingredients lead to the solution method summarized in Algorithm~\ref{algo:bisect},
which is a specialization of \cite[Algorithms~III and~IV]{LiZhCh12} to the case of $K=2$ users.
The algorithm uses
\begin{align}
\mb A_i &= \begin{bmatrix}
\mb\Psi & \mb\sigma\\
\frac{1}{P_i} \mb e_i^\Tr \mb\Psi & \frac{1}{P_i} \mb e_i^\Tr \mb\sigma
\end{bmatrix}, &&
\\
\mb \Psi &= \begin{bmatrix}
0 & d_1 |\mbc w_1^\He \mbc h_{12}|^2  \\
d_2 |\mbc w_2^\He \mbc h_{21}|^2 & 0 
\end{bmatrix},
&
\mb\sigma&= \begin{bmatrix}
d_1 \|\mbc w_1\|^2 \\  d_2 \|\mbc w_2\|^2
\end{bmatrix}
\end{align}
for $i\in\{1,2\}$, where $d_k = \frac{2^{\rho_k R}-1}{|\mbc w_k^\He \mbc h_{kk}|^2}$ for $k\in\{1,2\}$.

\begin{algorithm}[h]
\caption{Solution to Problem~\eqref{eq:noTS} based on \cite{LiZhCh12}}
\label{algo:bisect}
Perform a bisection to find $R:\Gamma(R)=1$ by repeatedly evaluating the nondecreasing function $\Gamma(R)$ as follows:

For $i\in\{1,2\}$:
\begin{enumerate}
\item Set $\covs{\mbc x_k} \gets 0,~\forall k$. \label{item:bisectFirst}
\item Set $\mbc w_k \gets \cov{\mbc s_k}^{-1}\mbc h_{kk},~\allk $. \label{item:bisectFilt}
\item Set $\lambda$ and $[\covs{\mbc x_1},\covs{\mbc x_2},1]$ to the dominant eigenvalue and eigenvector of $\mb A_i$.
\item Repeat from step~\ref{item:bisectFilt}) until $\lambda$ has converged (change from previous iteration smaller than some $\epsilon$).
\item If $\covs{\mbc x_k} \leq P_k,~\forall k$, return $\Gamma(R)\gets\frac{1}{\lambda}$.
\end{enumerate}
\end{algorithm}

\subsection{Coded Time-Sharing with Proper Signals}
\label{sec:algo:ts}
To solve the time-sharing problem \eqref{eq:primal} under a restriction to proper signals,
we extend the approach for the SISO interference channel in \cite{HeUt19b} to the considered SIMO interference channel.
Since \eqref{eq:primal} fulfills the so-called time-sharing condition from \cite{YuLu06}, it has zero duality gap, and we can
consider the Lagrangian dual problem (e.g., \cite{BoVa09,BaShSh06}).

Let $\mb\pow^{(\ms)}=[\pow_1^{(\ms)},\pow_2^{(\ms)}]^\Tr$, and define the rate with proper signals as $r_k(\mb\pow):=\left.r_k(\mathcal{X})\right|_{\mathcal{X} = (\pow_1,\pow_2,0,0)}$ with $r_k(\mathcal{X})$ from \eqref{eq:rk}.
We can drop the constraint \eqref{eq:primal:pcov} and dualize the constraints \eqref{eq:primal:rate}--\eqref{eq:primal:pow},
so that we obtain 
\begin{align}
\label{eq:dual_start}
\min_{\substack{\mb\mu\geq\zero\\\mb\lambda\geq\zero}} ~\max_{\substack{\Ms\in\mathbb{N},R\in\mathbb{R}\\(\mb\tau\geq\zero): \, \sum_{\ms=1}^{\Ms}\tau_\ms=1}} ~\max_{(\mb\pow^{(\ms)}\geq\zero)_{\forall\ms}} ~~ \Lag
\end{align}
with the dual variables $\mb\mu=[\mu_1,\mu_2]^\Tr$ and $\mb\lambda=[\lambda_1,\lambda_2]^\Tr$,
and the Lagrangian function
\begin{mueq}
\label{eq:lagrangian_reform}
\Lag = \left(1- \sum_{k=1}^2 \mu_k\rho_k \right)R ~+~ \sum_{k=1}^2 \lambda_k P_k 
\ifCLASSOPTIONdraftcls\else\\\fi
+\,\sum_{\ms=1}^\Ms \tau_\ms \sum_{k=1}^2\left( \mu_k \, r_k(\mb\pow^{(\ms)}) - \lambda_k \pow_k^{(\ms)} \right).
\end{mueq}
To avoid an unbounded inner maximization, the outer minimization must be restricted to $\mb\rho^\Tr\mb\mu= 1$,
and we obtain
\begin{equation}
\label{eq:dual_final}
\min_{\substack{\mb\mu\geq\zero,\mb\lambda\geq\zero\\\mb\rho^\Tr\mb\mu= 1}} ~~
\sum_{k=1}^2 \lambda_k P_k
+ f_{\mb\mu,\mb\lambda}(\mb\pow^\star(\mb\mu,\mb\lambda)).
\end{equation}
with
\begin{equation}
\label{eq:inner}
\mb\pow^\star(\mb\mu,\mb\lambda) = 
\argmax_{\mb\pow\geq\zero} ~ f_{\mb\mu,\mb\lambda}(\mb\pow)
\end{equation}
and
\begin{equation}
\label{eq:inner:func}
f_{\mb\mu,\mb\lambda}(\mb\pow)= \sum_{k=1}^2\left( \mu_k \, r_k(\mb\pow) - \lambda_k \pow_k \right).
\end{equation}
Note that the dual problem can finally be expressed without a dependence on $\Ms$ and $\tau_\ms$
since the $\Ms$ instances of the innermost maximization in \eqref{eq:dual_start} are all equivalent to solving the same inner problem \eqref{eq:inner}.
For a similar derivation with more details about intermediate steps, the reader is referred to \cite{HeUt19b}.

The outer problem is a convex program and can be solved by various methods from the literature on convex programming (e.g., \cite{BoVa09,BaShSh06}).
A detailed description how the problem in the SISO case can be solved by means of the cutting plane method \cite{Ke60,BaShSh06}
is given in \cite{HeUt19b}.
As switching to a SIMO scenario only changes the rate equations in the inner problem, but not the structure of the outer problem,
we refer the reader to \cite{HeUt19b} for further details and instead concentrate on the inner problem.

For given $\mb\mu$ and $\mb\lambda$, the inner problem \eqref{eq:inner} can be reformulated as a so-called \emph{mixed monotonic program} (MMP) \cite{MaHeJoUt20},
which can be solved by means of a branch-and-bound algorithm (e.g., \cite[Sec.~6.2]{Tu16}).
For this method, which we summarize below, we introduce the mixed monotonic (MM) function
\begin{mueq}
\brbfun(\brbx,\brby)=
\ifCLASSOPTIONdraftcls\else\\\fi
 \sum_{k=1}^2\left( \mu_k  \log_2\!\left(\!1+ \brbxk \mbc h_{kk}^\He \left(\id_{N_k} \!+ \mbc h_{k\notk} \brbyk[\notk] \mbc h_{k\notk}^\He\right)^{-1} \!\!\mbc h_{kk}\!\right)\!   - \lambda_k \brbyk \right)
\end{mueq}
which is nondecreasing in $\brbx$ and nonincreasing in $\brby$.
Since $f_{\mb\mu,\mb\lambda}(\mb\pow) = \brbfun(\mb\pow,\mb\pow)$, we can rewrite the inner problem as
\begin{equation}
\label{eq:mmp}
\max_{\mb\pow\geq\zero} ~ \brbfun(\mb\pow,\mb\pow)
\end{equation}
and we note that
\begin{subequations}
\label{eq:bounds}
\begin{align}
\label{eq:utop}
\brbU &:= \brbfun( \brbb,\brba) \geq \brbfun( \brbp,\brbp),~\forall \brbp\in\brbBox{\brba}{\brbb}
\\
\label{eq:cbv}
\brbL&:=\brbfun( \brba,\brba) 
\leq \max_{\brbp \in\brbBox{\brba}{\brbb}} \brbfun(\brbp,\brbp)
\end{align}
\end{subequations}
give us upper and lower bounds to the optimal value inside a box $\brbp\in\brbB=\brbBox{\brba}{\brbb}=\{\brbp~|~\brba\leq\brbp\leq\brbb\}$.
The utopian bound \eqref{eq:utop} is due to the MM properties of $\brbfun$ and becomes tight as $\brbb-\brba\to\zero$.

The branch-and-bound algorithm summarized in Algorithm~\ref{algo:bb} is based on the observation that subdividing a box $\hat\brbB=\brbBoxSmall{\hat\brba}{\hat\brbb}$ 
into a pair of smaller boxes $\brbB_1$ and $\brbB_2$ leads to refined bounds \eqref{eq:bounds},
which ultimately become tight if the boxes converge to singletons.
A new pair of boxes can be obtained by cutting the box $\hat\brbB$ along its longest edge into two subboxes, i.e.,
\begin{subequations}
\label{eq:branch}
\begin{align}
\brbB_1 &= \brbBox{\hat\brba}{\hat\brbb - \frac{\hatbrbbk[k^\star] -\hatbrbak[k^\star]}{2} \mb e_{k^\star}}
\label{eq:branch1}\\\label{eq:branch2}
\brbB_2 &= \brbBox{\hat\brba +  \frac{\hatbrbbk[k^\star] -\hatbrbak[k^\star]}{2} \mb e_{k^\star}}{\hat\brbb}
\end{align}
\end{subequations}
and
\begin{equation}
\label{eq:branch_k}
k^\star = \argmax_{k\in\{1,2\}} ~~\hatbrbbk-\hatbrbak.
\end{equation}

\begin{algorithm}[h]
\caption{Branch-and-Bound Method for Problem~\eqref{eq:inner}}
\label{algo:bb}
Given an initial set $\brbBset=\{\brbB_0\}$ such that the optimizer is contained in the box $\brbB_0$:
\begin{enumerate}
\item Find the box with the highest upper bound, i.e., $\hat\brbB = \argmax_{\brbB\in\brbBset} \brbU[\brbB]$ with $\brbUsymb$ defined in \eqref{eq:utop}.\label{item:brb_find}
\item Replace $\brbBset$ by\footnotemark $(\brbBset\setminus\{\hat\brbB\}) \cup \{\brbB_1,\brbB_2\}$ using \eqref{eq:branch}.\label{item:brb_branch}
\item Repeat steps~\ref{item:brb_find}) and~\ref{item:brb_branch}) until
$\max_{\brbB\in\brbBset} \brbU[\brbB]-\max_{\brbB\in\brbBset} \brbL[\brbB] \leq \brbeps$ with $\brbLsymb$ defined in \eqref{eq:cbv}.\label{item:brb_converge}
\item Return the vector $\brbp$ that achieves $\max_{\brbB\in\brbBset} \brbL[\brbB]$.
\end{enumerate}
\end{algorithm}
\footnotetext{We use $\setminus$ to denote a set difference.}

According to the convergence proof in \cite{MaHeJoUt20}, the method converges in finite time to an $\brbeps$-optimal solution,
i.e., a solution that is no more than $\brbeps$ away from the true global optimum,
if the initial box $\brbB_0$ contains the whole feasible set.
Since the feasible set of \eqref{eq:mmp} is unbounded, 
we instead use the following procedure (following the lines of \cite{HeUt19b}) to construct a $\brbB_0$ that contains the global optimum, which is sufficient.

Consider the concave interference-free expression
\begin{equation}
\brbinitfun_k(\brbpk) = 
\mu_k \log_2\left( 1+ \brbpk \|\mbc h_{kk}\|^2 \right) - \lambda_k \brbpk  ~\stackrel{\brbpk\to\infty}{\to}~ -\infty
\end{equation}
where setting the derivative to zero leads to the maximum
\begin{equation}
\brbinitfun_{\mt{max},k}=\max_{\brbpk\geq 0} \brbinitfun_k(\brbpk) = 
\brbinitfun_k\left(
\max\left\{\frac{\mu_k}{\lambda_k\ln 2} - \frac{1}{\|\mbc h_{kk}\|^2}  , 0\right\}
\right)\!.
\end{equation}
Moreover, by any root finding method for concave functions, we can find $\brbpk[0,k]$
such that $\brbinitfun_k(\brbpk)+\brbinitfun_{\mt{max},\notk}\leq0,~\forall \brbpk\geq\brbpk[0,k]$.
Since neglecting the interference cannot reduce the achievable rates, it holds that
$\brbinitfun(\brbp):=\sum_{k=1}^2 \brbinitfun_k(\brbpk)\geq\brbfun(\brbp,\brbp)$,
and we have $\brbinitfun(\brbp)\leq0$ if  $\brbpk\geq\brbpk[0,k]$ for any $k$.
Thus, \eqref{eq:mmp} takes its maximum inside $\brbB_0=\brbBox{\zero}{\brbp_0}$, where $\brbp_0=[\brbpk[0,1],\brbpk[0,2]]^\Tr$.

\begin{remark}
The algorithmic solution based on Lagrange duality, the cutting plane method, and mixed monotonic programming
could also be applied to the $K$-user SIMO interference channel with $K>2$ users.
The only obstacle is that the computational complexity of the branch-and-bound method grows exponentially in the number of variables,
and might thus no longer be feasibly if $K$ grows large.
Due to the nonconvex nature of the inner problem, polynomial-complexity methods for finding its global optimum are not expected to exist.
\end{remark}
\begin{remark}
Instead of the MMP approach, other monotonic programming formulations, e.g., based on the polyblock method as in \cite{LiZhCh12,Br12}, could be used.
However, the case studies in \cite{MaHeJoUt20} suggest that those methods would be computationally less efficient than the proposed MMP solution.
To get an overview of various monotonic programming techniques in similar scenarios, see \cite{ZhQiHu12,BjJo13,MaHeJoUt20} and the references therein.
\end{remark}

\begin{remark}
Note that solving the dual problem \eqref{eq:dual_final} only delivers the optimal rates, but does not directly deliver the time-sharing strategy that achieves these rates.
If we are interested in the optimal strategy, a so-called primal recovery as described in \cite{HeUt19b} can be easily performed based on the cutting plane solution
of the outer problem.
The number of strategies $\Ms$ obtained from the primal recovery can in principle be arbitrarily high,
as defined in the optimization problem \eqref{eq:primal}.
However, it is clear from an extension to the Carath\'eodory Theorem discussed in \cite{HaRa51} that
there always exists an optimal solution of \eqref{eq:primal} that requires no more than $4$ active strategies \cite{HeUt19b}.
Indeed, only a low number of strategies with nonzero time-sharing weights $\tau_\ms$ is observed when applying the algorithm in numerical simulations.
\end{remark}

\subsection{Weighted Sum Rate Maximization with Improper Signals}
\label{sec:algo:grad}
When improper signals are allowed as input signals,
it is no longer sufficient to consider the first summand of the rate equation \eqref{eq:rk}.
Instead, we have to consider the complete expression in \eqref{eq:rk} or the composite real version in \eqref{eq:rkR}.
As these expressions are nonconvex and do not have clear monotonicity properties,
it is not obvious how a globally optimal solution to rate maximization problems with improper signals can be obtained.
In the following, we propose a gradient-projection approach for the weighted sum rate maximization
\begin{equation}
\label{eq:wsrmax}
\max_{\substack{\covR{\crr x_1}\succeq\zero\\\covR{\crr x_2}\succeq\zero}} ~\underbrace{\sum_{k=1}^2 w_k r_k(\covR{\crr x_1},\covR{\crr x_2})}_{W(\covR{\crr x_1},\covR{\crr x_2})}
~~\st~~ \tr{\covR{\crr x_k}}\leq P_k,~\allk
\end{equation}
in the composite real representation.

The approach relies on the covariance-based optimization framework from \cite{HuScJoUt08},
but the following modifications are applied.
\begin{itemize}
\item In \cite[Th.~1]{HuScJoUt08}, a projection onto a feasible set defined by a sum power constraint is derived.
To apply the method to individual power constraints instead, note that the constraints on $\covR{\crr x_1}$ and $\covR{\crr x_2}$ in \eqref{eq:wsrmax} are not coupled.
Thus, we can project each $\covR{\crr x_k}$ individually on its respective feasible set.
Fortunately, this is equivalent to the projection with a sum power constraint in a system with $K=1$ user,
i.e., we can apply the projection from \cite[Th.~1]{HuScJoUt08} individually for each $\covR{\crr x_k}$.
\item While a convex optimization problem was considered in \cite{HuScJoUt08},
we apply the gradient-projection approach to a nonconvex problem.
Thus, we cannot expect to find the globally optimal solution in general.
To increase the probability of finding the global optimum,
we can start the algorithm with multiple random initializations and keep the best solution.
\item Instead of the complex formulation with proper signals in \cite{HuScJoUt08}, we consider a problem in the composite real representation.
Thus, all Hermitian matrices in \cite{HuScJoUt08} are replaced by real-symmetric matrices below.
Moreover, similar as for the composite real optimization methods in \cite{HeUt15a,LaAgVi16}, it is necessary
to choose initial covariance matrices that correspond to improper signaling
because the algorithm would otherwise be stuck in the set of solutions with proper signals and would never find potentially better solutions with improper signals.
\end{itemize}

The procedure summarized in Algorithm~\ref{algo:grad} makes use of the gradient
\begin{equation}
\label{eq:wsr:grad}
\frac{\partial W}{\partial \covR{\crr x_k}} = 
\frac{w_k}{2\ln2} \bsc H_{kk}^\Tr \covR{\crr y_k}^{-1} \bsc H_{kk} + 
\frac{w_\notk}{2\ln2} \bsc H_{\notk k}^\Tr \left( \covR{\crr y_\notk}^{-1} -\covR{\crr s_\notk}^{-1} \right) \bsc H_{kk}
\end{equation}
and of the projection
\begin{equation}
\proj{k}(\covR{\crr x_k}) := \mb\Omega_k \diag( \max\{\xi_{k,i} - \zeta_k,0 \}) \mb \Omega_k
\end{equation}
where $\mb\Omega_k \diag( \xi_{k,i}) \mb \Omega_k = \covR{\crr x_k}$ is an eigenvalue decomposition
and $\zeta_k$ is determined such that $\sum_{i=1}^{2N_k} \max\{\xi_{k,i} - \zeta_k,0 \} = P_k$.
Finding $\zeta_k$ has an interpretation similar to the waterfilling method \cite{CoTh06} and can, e.g., be implemented as described in \cite[Cor.~1]{HuScJoUt08}.

\begin{algorithm}[h]
\caption{Gradient-Projection Algorithm for Problem~\eqref{eq:wsrmax}}
\label{algo:grad}
For each given pair of initial matrices $(\covR{\crr x_1,0},\covR{\crr x_2,0})$:
\begin{enumerate}
\item Set $s\gets 1$.
\item Set $\mb G_k \gets \frac{\partial W}{\partial \covR{\crr x_k}},~\allk$ using \eqref{eq:wsr:grad}.
\label{item:grad:deriv}
\item Set $\covR{\crr x_k,\graditer+1}\gets \proj{k}\left(\covR{\crr x_k,\graditer}+\frac{1}{\gradss}\mb G_k \right),~\allk$.
\label{item:grad:step}
\item If $W(\covR{\crr x_1,\graditer+1},\covR{\crr x_2,\graditer+1})-W(\covR{\crr x_1,\graditer},\covR{\crr x_2,\graditer})<0$, set $\gradss\gets \gradss+1$ and repeat from step~\ref{item:grad:step}).
\label{item:grad:stepsize}
\item If $W(\covR{\crr x_1,\graditer+1},\covR{\crr x_2,\graditer+1})-W(\covR{\crr x_1,\graditer},\covR{\crr x_2,\graditer})>\gradeps$, 
set $\graditer\gets\graditer+1$ and repeat from step \ref{item:grad:deriv}).
\label{item:grad:converge}
\item Return $(\covR{\crr x_1,\graditer+1},\covR{\crr x_2,\graditer+1})$.
\end{enumerate}
\end{algorithm}

\begin{remark}
The gradient-projection approach could be further fine-tuned with a preconditioning step \cite{HuScJoUt08},
other step size rules (e.g., \cite{BaShSh06}),
or a different convergence criterion, e.g., based on the Frobenius norm of the change in the variables instead of based on the objective function.
However, for our purposes in Section~\ref{sec:num}, such a fine-tuning is not required.
Analyzing these possible modifications is thus beyond the scope of this paper.
\end{remark}

\begin{remark}
Even though we consider the SIMO interference channel, the composite real representation is a real-valued MIMO system.
Thus, Algorithm~\ref{algo:grad} is designed for weighed sum rate maximization in the real-valued MIMO interference channel.
This implies that it could also be applied to the composite real representation of a complex MIMO interference channel.
Moreover, the algorithm can easily be extended to a $K$-user MIMO interference channel with $K>2$ users
and is thus an alternative to the heuristic based on weighted minimum mean square error minimization from \cite{LaAgVi16}.
Since studying MIMO scenarios goes beyond the scope of this paper,
a comparison of the average performance of the two heuristics in various scenarios is left open for future research.
\end{remark}

\section{Numerical Results}
\label{sec:num}
To compare the various transmit strategies numerically, we consider the channel realization
\begin{subequations}
\label{eq:num:chan2}
\begin{align}
\mbc h_{11} &=
\begin{bmatrix}-0.0878 + 0.3457\J\\ \phantom{+} 1.0534 + 0.7316\J\end{bmatrix}\!, &
\mbc h_{12} &=
\begin{bmatrix} \phantom{+} 0.9963 + 0.5140\J\\ \phantom{+} 1.0021 - 0.2146\J\end{bmatrix}\!, \\
\mbc h_{21} &= 
\begin{bmatrix} \phantom{+} 0.9496 + 0.4156\J\\-1.7076 - 1.1134\J\end{bmatrix}\!, &
\mbc h_{22} &=
\begin{bmatrix} \phantom{+} 0.5072 + 0.6282\J\\ \phantom{+} 1.1528 - 0.8111\J\end{bmatrix}
\end{align}
\end{subequations}
with $P_1=P_2=10$.
According to the numerical results in Fig.~\ref{fig:num:chan2},
the convex hull formulation for proper signals cannot do any better than time division multiple access, i.e., switching between the two single-user points.
This result is obtained by first calculating the rate region for pure strategies with proper signals using the globally optimal method described in Section~\ref{sec:algo:pure}
and taking the convex hull afterwards.
By contrast, when incorporating the possibility of coded time-sharing directly in the optimization as discussed in Section~\ref{sec:algo:ts},
a significantly larger rate region can be achieved while sticking to proper signaling.
\begin{figure}
\ifCLASSOPTIONdraftcls\centering\fi
\hspace*{-4mm}
\begin{tikzpicture}

\begin{axis}[%
view={0}{90},
width=7.5cm,
height=6.25cm,
scale only axis,
xmin=0, xmax=6,
ymin=0, ymax=5,
xmajorgrids,ymajorgrids,
xlabel={$R_1$},
ylabel={\parbox{2cm}{\centering$R_2$\vspace*{-1.3cm}}},
%
legend style={nodes=right}
]

\addplot [color=black,solid]
  table[row sep=crcr]{%
0.00476470818811159	4.75994347992348\\
0.0626189522907549	4.58442759989885\\
0.118107239059179	4.43323145300936\\
0.171827821680848	4.29994372499053\\
0.224194847356812	4.1804190496336\\
0.275511561012442	4.07181090072234\\
0.326009402747565	3.9720710553614\\
0.375871049213782	3.87967437069883\\
0.425244422640948	3.79345024641608\\
0.47425181904396	3.71247798933353\\
0.522995531208423	3.63601465732775\\
0.571562496617703	3.56345285239805\\
0.620026780961733	3.49428497060915\\
0.668451848015569	3.42808161615216\\
0.716891572203376	3.36447234583463\\
0.765391096955861	3.303133006464\\
0.813987670585964	3.24377738717058\\
0.862711148560079	3.18615027844956\\
0.911583287710676	3.13001851991812\\
0.96061887182606	3.07517094784019\\
1.00982422941102	3.02141021933555\\
1.05919926157804	2.96855722317543\\
1.10873492685485	2.91644179071584\\
1.1584154040197	2.86490804486352\\
1.20821727119573	2.81381092852375\\
1.25811071390952	2.76301821588519\\
1.30805922842859	2.71240866285338\\
1.35802206465288	2.66187600284621\\
1.40795414364586	2.61132688303372\\
1.45746478850764	2.56007993896156\\
1.50746808804216	2.50976865557452\\
1.55733305636104	2.45925874545994\\
1.60702947493644	2.40853642945246\\
1.65644920283082	2.35748234822223\\
1.69776010065352	2.29556014079896\\
1.73589211266599	2.23072796769046\\
1.77332269601508	2.16651888588912\\
1.8101109684159	2.10300791698842\\
1.8463303878805	2.04027183900193\\
1.88206190919859	1.97837993316557\\
1.91739291082309	1.91739291082309\\
1.9524198107208	1.85736566314778\\
1.98729727701424	1.7983914114594\\
2.02243283914166	1.74075800449542\\
2.05751280057933	1.68409985732089\\
2.09261201358325	1.62841401634918\\
2.12780336623437	1.57368983413855\\
2.16316033135372	1.51991178604929\\
2.19875348049002	1.46705759068374\\
2.23465316670592	1.4151008926725\\
2.27092925683588	1.36401153041643\\
2.30765223137488	1.31375658242696\\
2.34489165620244	1.26429974937312\\
2.38271875719884	1.21560307192313\\
2.42120714422181	1.16762727988226\\
2.46043140331451	1.12033105375591\\
2.50046928202261	1.07367205877586\\
2.54140564395526	1.02760835591169\\
2.58332593960795	0.98209527301775\\
2.62632556343772	0.937088925131489\\
2.67050650172186	0.892544201042462\\
2.71598012029465	0.84841519489896\\
2.76286770003155	0.804654670723891\\
2.81130688355746	0.761215127507713\\
2.86144839004097	0.718046718841904\\
2.91346452566275	0.675098400480744\\
2.96754868614655	0.632316281581419\\
3.02392547475311	0.589644237854794\\
3.08285342093512	0.54702226602487\\
3.14463240323104	0.50438549962184\\
3.20961999762745	0.461663957336748\\
3.2782395151864	0.418779870966465\\
3.35100291545384	0.375646234294648\\
3.42853923964693	0.332164124651674\\
3.51162893845914	0.288218422314695\\
3.60126155612121	0.243672709055579\\
3.69871975711524	0.198361432554173\\
3.80571342663646	0.152078140985889\\
3.92458564386251	0.104556231670071\\
4.05869778142127	0.0554379793767534\\
4.21311984207485	0.00421733717925411\\
};
\addlegendentry{proper pure}

\addplot [color=black,dashed]
  table[row sep=crcr]{%
0.00476470818811159	4.75994347992348\\
4.21311984207485	0.00421733717925411\\
};
\addlegendentry{proper cvx.\ hull}

\addplot [color=gray,dashed,mark=triangle*,mark options={solid}]
  table[row sep=crcr]{%
0	4.77534426964198\\
0	4.77534426964198\\
0	4.77534426964198\\
0	4.77534426964198\\
0	4.77534426964181\\
0	4.77534426964196\\
1.86352931450763	4.05026764963119\\
1.86382714228471	4.05012332598061\\
2.82817592440046	3.41026404032461\\
2.82817593188175	3.41026403834125\\
2.82817593451567	3.41026403763872\\
2.82817593527918	3.41026403702216\\
2.82817593578122	3.41026403523053\\
2.82817593686417	3.41026402817407\\
3.5030860688231	2.17629484905353\\
3.50338369576025	2.1755160744454\\
3.5035744111469	2.17486519029129\\
4.22659234751577	0\\
4.22659234751577	0\\
4.22659234751577	0\\
4.22659234751575	0\\
};
\addlegendentry{improper heuristic}

\addplot [color=black,dashdotted]
  table[row sep=crcr]{%
0.00477920412558983	4.77442492146424\\
0.252931484547146	4.71625288769542\\
0.520879515085279	4.64657599171312\\
0.806563917991503	4.54555232614588\\
1.11010674190697	4.42382517188651\\
1.42969457064143	4.27766898481337\\
1.8908232108468	4.03466475796173\\
2.10505581506639	3.90423854709858\\
2.45345608456479	3.67711883938521\\
2.82422424297675	3.41026635962093\\
3.04664756017678	3.04664756017678\\
3.25786269068267	2.66660119490138\\
3.49913094770218	2.18127643493123\\
3.60741441541974	1.94501657645303\\
3.75020444443393	1.61029361793589\\
3.87071129372272	1.29368002545369\\
3.9701279707849	0.99625678126026\\
4.06124626864138	0.720628532537685\\
4.12980122057516	0.462949247146327\\
4.18660972614511	0.224526851818339\\
4.22604536777152	0.00423027564341494\\
};
\addlegendentry{proper time-sharing\!}

\end{axis}
\end{tikzpicture}%
\vspace*{-5mm}\caption{Achievable rate regions with pure strategies, with the convex hull formulation, and with time-sharing in the scenario \eqref{eq:num:chan2} with $P_1=P_2=10$.}
\label{fig:num:chan2}
\end{figure}
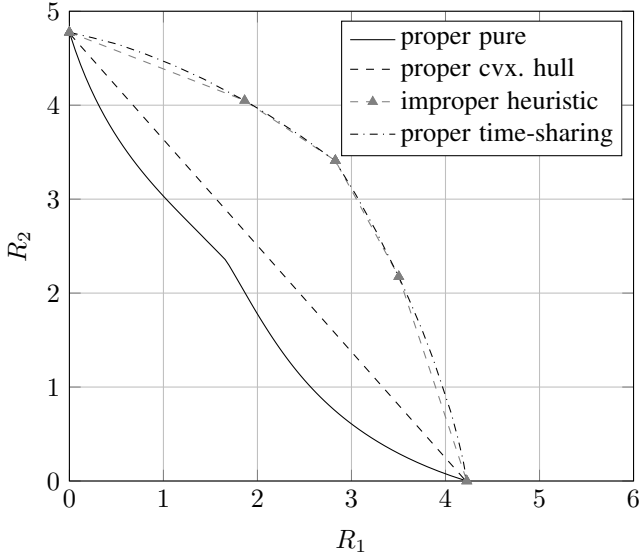

For improper signaling, we present results based on the gradient-projection heuristic for weighted sum rate maximization from Section~\ref{sec:algo:grad}.
When solved optimally, a weighted sum rate maximization can only lead to points that are achievable with pure strategies, but also lie on the Pareto boundary of the convex hull of the rate region \cite[Cor~A5.9]{Br12}.
The results can thus be used to draw conclusions about both pure strategies and the convex hull formulation.

For a better illustration, we have added markers at the points that can be achieved with a pure strategy.
As some of these points lie outside the rate region of pure strategies with proper signals,
we have demonstrated a gain by improper signals for the case of pure strategies.
A gain for the case of the convex hull formulation is demonstrated by the convex hull of these points (gray dashed) compared to the convex hull for proper signaling.
However, the gains by means of improper signaling are limited to these two cases without coded time-sharing.
If coded time-sharing is allowed, Theorem~\ref{th:main} states that no gain by using improper signaling instead of proper signaling is possible.

As we have solved the weighted sum rate maximization with a suboptimal heuristic, it is left open whether a better algorithm could bring
the convex hull for improper signals closer to the plotted time-sharing rate region.
However, it is clear from Theorem~\ref{th:main} that the time-sharing rate region is an outer bound to what is achievable with improper signaling.

In a second scenario with 
\begin{subequations}
\label{eq:num:chan1}
\begin{align}
\mbc h_{11} &=
\begin{bmatrix} \phantom{+} 0.9578 + 2.0563\J\\-0.7581 + 0.5835\J\end{bmatrix}\!, &
\mbc h_{12} &=
\begin{bmatrix} \phantom{+} 0.6795 + 0.9751\J\\ \phantom{+} 0.0877 - 0.7482\J\end{bmatrix}\!, \\
\mbc h_{21} &= 
\begin{bmatrix} \phantom{+} 1.0159 - 0.3314\J\\-1.3866 - 0.1927\J\end{bmatrix}\!, &
\mbc h_{22} &=
\begin{bmatrix} -0.1398 + 0.7767\J\\-0.8541 - 0.1965\J\end{bmatrix}
\end{align}
\end{subequations}
we observe a case where proper pure strategies and the convex hull formulation with proper signals achieve more than time division multiple access.
All other observations remain the same as in the first scenario.
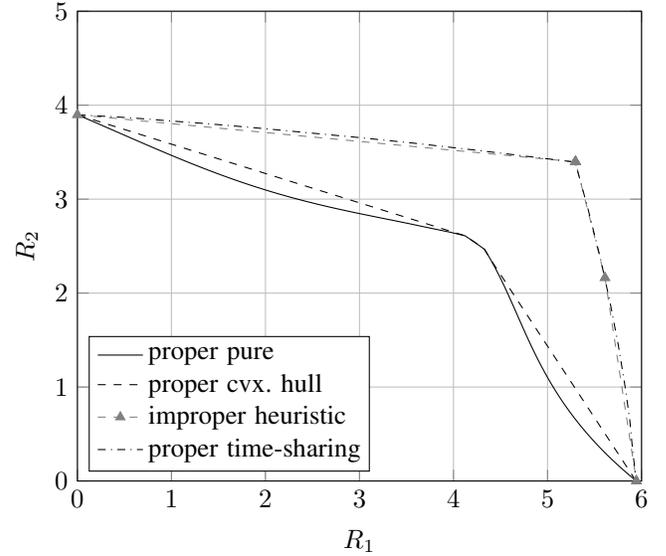
\begin{figure}
\ifCLASSOPTIONdraftcls\centering\fi
\hspace*{-4mm}
\begin{tikzpicture}

\begin{axis}[%
view={0}{90},
width=7.5cm,
height=6.25cm,
scale only axis,
xmin=0, xmax=6,
ymin=0, ymax=5,
xmajorgrids,ymajorgrids,
xlabel={$R_1$},
ylabel={\parbox{2cm}{\centering$R_2$\vspace*{-1.3cm}}},
legend style={at={(0.02,0.02)},anchor=south west,nodes=right}
]

\addplot [color=black,solid]
  table[row sep=crcr]{%
0.00390030853960888	3.89640823106927\\
0.0529343155782923	3.87540079227271\\
0.10267550239878	3.85399125670641\\
0.153135872873534	3.83218287464849\\
0.204328251783	3.80997924886532\\
0.256266432370176	3.78738536045311\\
0.308965233719689	3.76440633806264\\
0.362440884869041	3.74105059397659\\
0.416710844457608	3.71732531087581\\
0.471794332392469	3.69324060375822\\
0.527712475056174	3.66880820133487\\
0.584488038300445	3.64403819284765\\
0.64214666526183	3.61894600402702\\
0.700716048413335	3.59354501126698\\
0.760227166206691	3.56785234536001\\
0.82071383923523	3.54188464174938\\
0.882213582764903	3.51566070818676\\
0.944767722665821	3.48920023540444\\
1.00842195283412	3.4625244130897\\
1.07322698044253	3.43565645803043\\
1.1392383772124	3.40861941605068\\
1.20651748932382	3.381437598893\\
1.27513252298304	3.35413784544333\\
1.34515811639696	3.32674642956801\\
1.41667711567888	3.29929197779275\\
1.48978049931112	3.27180319804764\\
1.56456827224366	3.24430915896476\\
1.64115124062646	3.21684103533434\\
1.71965017274324	3.1894282535863\\
1.80019886541841	3.16210246576045\\
1.88294289391012	3.1348929326325\\
1.96804156230003	3.1078280935128\\
2.05566964777616	3.08093616875598\\
2.14601593095412	3.05424076245866\\
2.2392837814018	3.02776027693479\\
2.33569328368665	3.00150930685697\\
2.43547728916145	2.97549202690486\\
2.53887830564299	2.94969826170931\\
2.64614708980604	2.92410254666599\\
2.75753191096008	2.89865374302706\\
2.87326975409803	2.87326975409803\\
2.99356564494852	2.84782299829949\\
3.11857303601306	2.82213186162111\\
3.24835379695308	2.79593851723603\\
3.38284129306252	2.76889773778404\\
3.5217813858801	2.74055493042148\\
3.66467614246947	2.71033671735391\\
3.81315400014119	2.67925480272126\\
3.96793139935455	2.6474927409498\\
4.12208165147667	2.61032070281947\\
4.22779170903471	2.53938189486238\\
4.33036453490974	2.46529560853997\\
4.38941013433808	2.36664671395818\\
4.44406673043875	2.26725086753573\\
4.49495590752495	2.16769273625323\\
4.54273928510284	2.06848761699334\\
4.58807728452371	1.97006634687096\\
4.63159537750372	1.87276915924272\\
4.67387370417828	1.77685254546395\\
4.71543296838106	1.68249514583631\\
4.75673730733108	1.5898101340713\\
4.79819293850154	1.49885478419479\\
4.84015741739312	1.40964233390537\\
4.88294216585083	1.32215001682329\\
4.92682141891666	1.23632771658079\\
4.97203582055555	1.15210375826576\\
5.01880059116768	1.06939082166386\\
5.06731180326332	0.988090226149424\\
5.11774746820376	0.908094364133176\\
5.1702779872599	0.829290330758029\\
5.22506432067416	0.751560581440979\\
5.28226647369111	0.674785006408255\\
5.34204749745717	0.598841623380431\\
5.40457784123517	0.523606918942711\\
5.4700378871454	0.4489556605962\\
5.53862835378974	0.37476105369964\\
5.61057146730103	0.300893570420001\\
5.686123415477	0.227220229560568\\
5.76558539430813	0.153602937202706\\
5.84931882988916	0.0798961721525116\\
5.93777446089725	0.00594371817907633\\
};
\addlegendentry{proper pure}

\addplot [color=black,dashed]
  table[row sep=crcr]{%
0.00390030853960888	3.89640823106927\\
4.12208165147667	2.61032070281947\\
4.22779170903471	2.53938189486238\\
4.33036453490974	2.46529560853997\\
4.38941013433808	2.36664671395818\\
5.93777446089725	0.00594371817907633\\
};
\addlegendentry{proper cvx.\ hull}

\addplot [color=gray,dashed,mark=triangle*,mark options={solid}]
  table[row sep=crcr]{%
0	3.89807518349661\\
0	3.89807518349664\\
5.29705147467635	3.39704979783487\\
5.29705155210392	3.39704979587089\\
5.29705157863657	3.39704979584532\\
5.29705157811175	3.39704979555203\\
5.29705158095144	3.39704979566367\\
5.2970515817128	3.39704979571116\\
5.29705158231161	3.39704979574743\\
5.29705158230679	3.39704979574774\\
5.29705158231195	3.39704979574138\\
5.29705158232735	3.39704979462853\\
5.2970515822153	3.39704979411936\\
5.29705158226833	3.39704979411402\\
5.297051582351	3.39704979199761\\
5.29705158231157	3.39704979574822\\
5.61093834316287	2.16346693648015\\
5.61093944715364	2.16346178069701\\
5.94508556879281	0\\
5.94508556879281	0\\
5.94508556879281	0\\
};
\addlegendentry{improper heuristic}

\addplot [color=black,dashdotted]
  table[row sep=crcr]{%
0.0039013976031009	3.8974962054978\\
0.208422969584053	3.88633085328536\\
0.434277211130913	3.87402845485036\\
0.684345092723946	3.85676368447543\\
0.962143460381083	3.83418485657347\\
1.2737937535326	3.81121125058954\\
1.62436526716331	3.78297583524451\\
2.02052687605607	3.74746306415538\\
2.47193298701806	3.70481110847934\\
2.992393518259	3.65589245876611\\
3.59405774404157	3.59405774404157\\
4.29403497508661	3.51472111708762\\
5.29327997428053	3.39419741662794\\
5.4234429329035	2.92416816899507\\
5.55894771596278	2.38694667506463\\
5.66647576693425	1.8938654831448\\
5.75334349672959	1.44373368206649\\
5.82067515999008	1.03282202591605\\
5.87394985716313	0.658467688614373\\
5.90651069007239	0.316764718285413\\
5.94384661957451	0.00594979641599051\\
};
\addlegendentry{proper time-sharing\!}

\end{axis}
\end{tikzpicture}%
\vspace*{-5mm}\caption{Achievable rate regions with pure strategies, with the convex hull formulation, and with time-sharing in the scenario \eqref{eq:num:chan1} with $P_1=P_2=10$.}
\label{fig:num:chan1}
\end{figure}

As a further example, we reconsider the channel realization from \eqref{eq:num:chan2}, but we set $\mbc h_{12} = \zero$,
so that receiver~$1$ does not experience any interference from transmitter~2.
The results in Fig.~\ref{fig:num:zifc} show that the previous observations remain valid in such a one-sided SIMO interference channel (SIMO Z-interference channel).
\begin{figure}
\ifCLASSOPTIONdraftcls\centering\fi
\hspace*{-4mm}
\begin{tikzpicture}

\begin{axis}[%
view={0}{90},
width=7.5cm,
height=6.25cm,
scale only axis,
xmin=0, xmax=6,
ymin=0, ymax=5,
xmajorgrids,ymajorgrids,
xlabel={$R_1$},
ylabel={\parbox{2cm}{\centering$R_2$\vspace*{-1.3cm}}},
legend style={at={(0.02,0.02)},anchor=south west,nodes=right}
]

\addplot [color=black,solid]
  table[row sep=crcr]{%
0	4.77534426964198\\
0.230110740872596	4.3721042218333\\
0.450115898950419	4.05104297896738\\
0.666904028285173	3.77912314723907\\
0.88504723494064	3.54018927532295\\
1.10837967614518	3.32513887341707\\
1.34074562240985	3.12840645217319\\
1.58655824258013	2.94646536511531\\
1.85140387407809	2.77710579680752\\
2.14287385019784	2.61906802714631\\
2.47174126705457	2.47174126705458\\
2.85312588682085	2.3343757425128\\
3.30470569844543	2.20313708984543\\
3.83195268064167	2.0633591850472\\
4.22527857343002	1.81083355981539\\
4.22582957232771	1.40860982469878\\
4.22611851368101	1.05652967201668\\
4.22630142179427	0.745817805985604\\
4.22642848697098	0.469603094748196\\
4.22652157933875	0.222448511001337\\
4.22659234751577	0\\
};
\addlegendentry{proper pure}

\addplot [color=black,dashed]
  table[row sep=crcr]{%
0	4.77534426964198\\
4.22527857343002	1.81083355981539\\
4.22582957232771	1.40860982469878\\
4.22611851368101	1.05652967201668\\
4.22630142179427	0.745817805985604\\
4.22642848697098	0.469603094748196\\
4.22652157933875	0.222448511001337\\
4.22659234751577	0\\
};
\addlegendentry{proper cvx.\ hull}

\addplot [color=gray,dashed,mark=triangle*,mark options={solid}]
  table[row sep=crcr]{%
0	4.77534426964198\\
0	4.77534426964198\\
0	4.77534426964198\\
2.59376858591119	4.01420079498772\\
4.22659234751564	3.27626740281447\\
4.22659234751574	3.27626740283716\\
4.22659234748508	3.27626740284489\\
4.22659234748209	3.27626740284511\\
4.22659234751123	3.27626740284296\\
4.22659234751572	3.27626740282604\\
4.22659234749457	3.27626740284419\\
4.22659234750941	3.27626740284309\\
4.22659234751022	3.27626740284303\\
4.22659234751577	3.276267402811\\
4.22659234751561	3.27626737861387\\
4.22659234751573	3.27626738539807\\
4.22659234751577	3.27626740100322\\
4.22659234751576	3.27626740284262\\
4.22659234745594	3.27626740284633\\
4.22659234751283	0.902147710254411\\
};
\addlegendentry{improper heuristic}

\addplot [color=black,dashdotted]
  table[row sep=crcr]{%
0.00477941078675963	4.77463137597287\\
0.253818175991781	4.73278646038899\\
0.523922016269331	4.67371703402165\\
0.816648533972944	4.60238619710167\\
1.13270703520708	4.51388835465871\\
1.47448653987362	4.41168727199711\\
1.84243296315525	4.29083255999805\\
2.23653055877107	4.14808422504587\\
2.65994821488195	3.98659904869114\\
3.10986964740723	3.79941639437733\\
3.5857517124036	3.5857517124036\\
4.22628875180431	3.26167293120443\\
4.22635498804453	3.24253128694063\\
4.22635498804453	0\\
};
\addlegendentry{proper time-sharing\!}

\end{axis}
\end{tikzpicture}%
\vspace*{-5mm}\caption{Achievable rate regions with pure strategies, with the convex hull formulation, and with time-sharing in the scenario \eqref{eq:num:chan2} with $P_1=P_2=10$
and with $\mbc h_{12}$ replaced by $\mbc h_{12} = \zero$.}
\label{fig:num:zifc}
\end{figure}
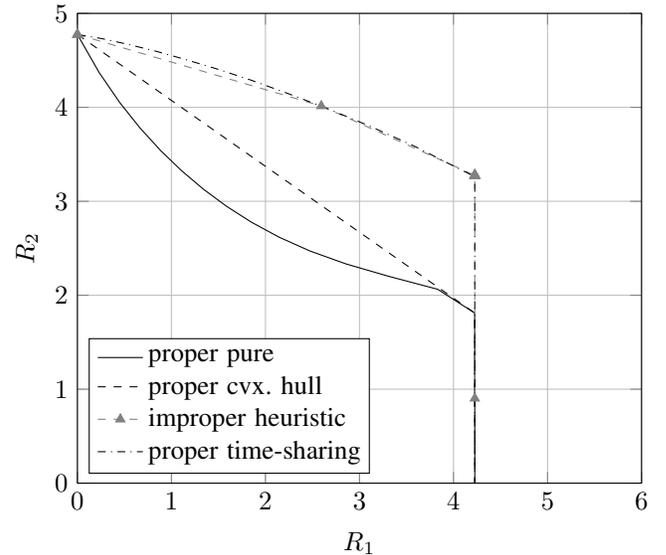

\section{Discussion and Outlook}
\label{sec:conclusion}
We have considered the two-user Gaussian SIMO interference channel with Gaussian inputs and interference treated as noise.
For this scenario, we have proven that proper signals achieve the whole rate region if coded time-sharing is allowed.
On the other hand, for the case where time-sharing is not allowed and only pure strategies or the convex hull formulation are considered,
we have demonstrated numerically that improper signaling can lead to larger rate regions than the globally optimal proper signaling strategy.
These results were established for two-user SISO interference channels in 
\cite{HeUt19b} and \cite{ZeYeGuGuZh13}, respectively,
and we have extended them to the two-user SIMO interference channel.

It is important to note that Theorem~\ref{th:main}, which shows the optimality of proper signals in case of coded time-sharing,
is specific to scenarios with only two users.
Improper signaling may bring gains in the $K$-user SIMO interference channel with $K>2$ even if coded time-sharing is allowed
since the phases of the complex pseudovariances play a role in systems with three or more users \cite{CaJaWa10}.
By contrast, the proof in this paper makes use of the fact that the pseudovariances can be chosen to be real-valued in the case of two users
(see the discussion below Lemma~\ref{lem:proper_opt}).

As a real-valued pseudovariance corresponds to a power imbalance between real and imaginary parts,
the following intuitive interpretation is possible in the two-user case.
In the convex hull formulation, 
reducing the transmit power of a user in one time slot will help the other user by reducing the interference,
but it will not allow us to use a higher transmit power in some other time slot.
The latter would be possible in case of coded time-sharing.
By contrast, reducing the power in one of the real-valued components of a complex signal
always allows us to increase the power in the other component without increasing the total transmit power.
Thus, improper signaling brings a gain in flexibility over proper signaling in the convex hull formulation,
but not in case of coded time-sharing.
For a more detailed discussion of this interpretation, the reader is referred to \cite{HeUt17a,HeUt19b},
where the aspect was discussed in single-antenna scenarios.

A topic for future research would be an extension of Theorem~\ref{th:main} to the MISO interference channel.
For this scenario, gains by means of improper signals were shown in \cite{ZeZhGuGu13} for pure strategies and for the convex hull formulation,
but it is an open question whether improper signals can still be beneficial if coded time-sharing is considered.
Extending the argumentation from this paper to the MISO interference channel is nontrivial since covariance matrices and pseudocovariance matrices of the input signals need to be considered
instead of scalar variances and scalar pseudovariances.

The situation is even less clear for the MIMO interference channel.
When studying the case of coded time-sharing, the complications are the same as described above for the MISO scenario.
However, for the MIMO interference channel, even the case of pure strategies and the convex hull formulation is not fully understood.
Since no globally optimal methods for pure strategies with proper signaling are available if all terminals have multiple antennas,
previous comparisons \cite{LaAgVi16} have only compared proper heuristics to improper heuristics,
but gains by improper signaling over globally optimal proper signaling have not been demonstrated.

\appendix

In this appendix, we prove Lemma~\ref{lem:trans},
and we provide formal proofs of Lemmas~\ref{lem:enhanced} and~\ref{lem:proper_opt} taking into account the possibility of symbol extensions.

\begin{proof}[Proof of Lemma~\ref{lem:trans}]
In the null space of $\mbc Q_k^\He$, we only have noise,
and due to $\cov{\bnoise_k}=\id_{N_k}$, these noise components are independent of the noise components in the orthogonal complement of the null space.
Thus, we do not lose any information when removing all components in this null space (see \cite[Sec.~2.2.1]{VeGa18}) by setting
\begin{align}
\mbc y_k'&=\mbc Q_k^\He \mbc y_k
=\mbc Q_k^\He
\begin{bmatrix}
\mbc h_{kk} & \mbc h_{k\notk}
\end{bmatrix}
\begin{bmatrix}
\mc x_k \\ \mc x_\notk
\end{bmatrix}+\mbc Q_k^\He \bnoise_k
\nonumber\\&=
\underbrace{\mbc Q_k^\He\mbc Q_k}_{=\id_2}\begin{bmatrix}
h_k & a_k \E^{\J\varphi_k} \\
0 & b_k \E^{\J\psi_k}
\end{bmatrix}
\begin{bmatrix}
\mc x_k \\ \mc x_\notk
\end{bmatrix}
+\mbc Q_k^\He \bnoise_k.
\end{align}
Furthermore, defining new input and output signals with rotated phases according to
\begin{subequations}
\label{eq:TransProofPhases}
\begin{align}
\mc x_1' &= \mc x_1  &
\mc x_2' &= \mc x_2  \E^{\J \varphi_1} \\
\mc y_{1,1}'' &= \mc y_{1,1}' &
\mc y_{2,1}'' &= \mc y_{2,1}' \E^{-\J \varphi_2}  &
\\
\mc y_{1,2}'' &= \mc y_{1,2}' \E^{\J (\varphi_1-\psi_1)} &
\mc y_{2,2}'' &= \mc y_{2,2}' \E^{-\J \psi_2} 
\end{align}
\end{subequations}
does not change the achievable rates.\footnote{On the receiver side, we do not lose information since the transformation is invertible.
On the transmitter side, for any distribution of $\mc x_2$ that fulfills the constraints, there exists a distribution for $\mc x_2'$ that fulfills the constraints.}
Thus, the SIMO interference channel
\begin{subequations}
\label{eq:modelTransProof}
\begin{align}
\label{eq:modelTransProof1}
\mbc y_1'' &= \begin{bmatrix}
h_1 & a_1 \E^{\J\varphi_1} \\
0 & \E^{\J (\varphi_1-\psi_1)}b_1 \E^{\J\psi_1}
\end{bmatrix}
\begin{bmatrix}
\mc x_1' \\  \E^{-\J\varphi_1} \mc x_2'
\end{bmatrix}
 + 
\bnoise_1'\\
\label{eq:modelTransProof2}
\mbc y_2'' &= 
\begin{bmatrix}
\E^{-\J \varphi_2}h_2 & \E^{-\J \varphi_2}a_2 \E^{\J\varphi_2} \\
0 & \E^{-\J \psi_2} b_2 \E^{\J\psi_2}
\end{bmatrix}
\begin{bmatrix}
 \E^{-\J\varphi_1} \mc x_2' \\ \mc x_1'
\end{bmatrix}
+ \bnoise_2'
\end{align}
\end{subequations}
with
\begin{subequations}
\begin{align}
\bnoise_1'&=\begin{bmatrix}
1 &~\\~&\E^{-\J(\varphi_1-\psi_1)} 
\end{bmatrix}
\mbc Q_1^\He \bnoise_1 \sim\mathcal{CN}(\zero,\id_2)
\\
\bnoise_2'&=\begin{bmatrix}
\E^{-\J\varphi_2}  &~\\~&\E^{-\J\psi_2} 
\end{bmatrix}
\mbc Q_1^\He \bnoise_2\sim\mathcal{CN}(\zero,\id_2)
\end{align}
\end{subequations}
has the same achievable rate region as \eqref{eq:model}.
Comparing \eqref{eq:modelTransProof} to \eqref{eq:modelTrans},
we can identify $\theta=-\varphi_1-\varphi_2$.
Moreover, $\mc x_1'$ and $\mc x_2'$ are proper if and only if $\mc x_1$ and $\mc x_2$ are proper
since $\pcovs{\mc x_2}=0~\Leftrightarrow~\pcovs{\mc x_2'}=\pcovs{\mc x_2}\E^{\J2\varphi_1}=0~$.
Finally, it is clear from \eqref{eq:TransProofPhases} that a strategy without symbol extensions in the transformed system corresponds to a strategy without symbol extensions
in the original system.
\end{proof}

\begin{proof}[Proof of Lemma~\ref{lem:enhanced}]
We can calculate the achievable rates with symbol extensions \eqref{eq:symext} and possibly improper signals by applying the rate equation \eqref{eq:rkRT} to the transformed SIMO interference channel
\eqref{eq:modelTrans}.
We use the eigenvalue decomposition $\covR{\ubar{\crr x}_\notk}=\mb V_k \mb\Phi_k \mb V_k^\Tr$
and we note that 
$\ubar{\bsc H}_{k\notk}^{\prime\Tr}\ubar{\bsc H}_{k\notk}'$ is equal to the composite real representation of the complex matrix
\begin{align}
\label{eq:enhanced:prodHkj}
\ubar{\mb H}_{k\notk}^{\prime\Tr}\ubar{\mb H}_{k\notk}' =
(\id_T \kron \mb h_{k\notk}')^\Tr(\id_T \kron \mb h_{k\notk}')
= \|\mb h_{k\notk}'\|^2\, \id_{T}.
\end{align}
Translating this to the composite real representation \eqref{eq:crrmat}, where the dimension is doubled, we have $\ubar{\bsc H}_{k\notk}^{\prime\Tr}\ubar{\bsc H}_{k\notk}'= \|\mb h_{k\notk}'\|^2\, \id_{2T}$, and
the denominator of \eqref{eq:rkRT} can be rewritten as
\begin{subequations}
\begin{align}
\label{eq:MIMORdenom}
\det \covR{\ubar{\crr s}_k}
&=2^{-4T}
\det\left(
 \id_{4T}+2\ubar{\bsc H}_{k\notk}' \covR{\ubar{\crr x}_\notk} \ubar{\bsc H}_{k\notk}^{\prime\Tr}\right)
\\
&=2^{-4T}
\det\left(
\id_{2T} + 2 \|\mb h_{k\notk}'\|^2\, \mb\Phi_\notk \right).
\end{align}
\end{subequations}

Furthermore, $\ubar{\bsc H}_{kk}^{\prime\Tr}\ubar{\bsc H}_{k\notk}'$ is equal to the composite real representation of\footnote{We treat $\mbc h_{kk}'$ as a complex vector since it is complex for $k=2$, see \eqref{eq:modelTrans:chan}.}
\begin{align}
\ubar{\mbc H}_{kk}^{\prime\He}\ubar{\mb H}_{k\notk}' =
(\id_T \kron \mbc h_{kk}')^\He(\id_T \kron \mb h_{k\notk}')
=  \mbc h_{kk}^{\prime\He}  \mb h_{k\notk}'\id_{T}.
\end{align}
Thus, we can write $\ubar{\bsc H}_{kk}^{\prime\Tr}\ubar{\bsc H}_{k\notk}'=|\mbc h_{kk}^{\prime\He}\mb h_{k\notk}'|\,\mb U_k$
with the real unitary matrix
\begin{equation}
\label{eq:enhanced:Uk}
\mb U_k = \begin{bmatrix}
\cos(\angle(\mbc h_{kk}^{\prime\He}\mb h_{k\notk}')) & -\sin(\angle(\mbc h_{kk}^{\prime\He}\mb h_{k\notk}')) \\
\sin(\angle(\mbc h_{kk}^{\prime\He}\mb h_{k\notk}')) &  \cos(\angle(\mbc h_{kk}^{\prime\He}\mb h_{k\notk}')) 
\end{bmatrix} \kron \id_{T}.
\end{equation}
Moreover, we have
$\ubar{\bsc H}_{kk}^{\prime\Tr}\ubar{\bsc H}_{kk}'= \|\mbc h_{kk}'\|^2\, \id_{2T}$
in analogy to \eqref{eq:enhanced:prodHkj}.
This enables us to rewrite the numerator of \eqref{eq:rkRT} as\ifCLASSOPTIONdraftcls\else\ given in\fi
\ifCLASSOPTIONdraftcls\else\begin{figure*}[t!]\fi
\begin{subequations}
\label{eq:MIMORnum}
\begin{align}
\label{eq:MIMORnum1}
\det \covR{\ubar{\crr y}_k}
&=2^{-4T}
\det\left(
 \id_{4T}+2 \begin{bmatrix} \ubar{\bsc H}_{kk}' & \ubar{\bsc H}_{k\notk}' \end{bmatrix}
 \begin{bmatrix}  \covR{\ubar{\crr x}_k} & ~ \\ ~ &  \covR{\ubar{\crr x}_\notk} \end{bmatrix}
 \begin{bmatrix} \ubar{\bsc H}_{kk}^{\prime\Tr} \\ \ubar{\bsc H}_{k\notk}^{\prime\Tr} \end{bmatrix}
 \right)
\\
\label{eq:MIMORnum2}
&=2^{-4T}
\det\left(
 \id_{4T}+2 \begin{bmatrix} \ubar{\bsc H}_{kk}' \mb V_k & \ubar{\bsc H}_{k\notk}' \mb V_\notk \end{bmatrix}
 \begin{bmatrix} \mb\Phi_k & ~ \\ ~ &  \mb\Phi_\notk \end{bmatrix}
 \begin{bmatrix} \mb V_k^\Tr \ubar{\bsc H}_{kk}^{\prime\Tr} \\ \mb V_\notk^\Tr \ubar{\bsc H}_{k\notk}^{\prime\Tr} \end{bmatrix}
 \right)
\\
\label{eq:MIMORnum3}
&=2^{-4T}
\det\left(
 \id_{4T}+2 \begin{bmatrix}  \mb V_k^\Tr \ubar{\bsc H}_{kk}^{\prime\Tr}\ubar{\bsc H}_{kk}' \mb V_k &  \mb V_k^\Tr \ubar{\bsc H}_{kk}^{\prime\Tr}\ubar{\bsc H}_{k\notk}' \mb V_\notk \\
  \mb V_\notk^\Tr \ubar{\bsc H}_{k\notk}^{\prime\Tr}\ubar{\bsc H}_{kk}' \mb V_k & \mb V_\notk^\Tr \ubar{\bsc H}_{k\notk}^{\prime\Tr}\ubar{\bsc H}_{k\notk}' \mb V_\notk \end{bmatrix}
 \begin{bmatrix} \mb\Phi_k & ~ \\ ~ &  \mb\Phi_\notk \end{bmatrix}
 \right)
\\
\label{eq:MIMORnum4}
&=2^{-4T}
\det\left(
 \id_{4T}+2 \begin{bmatrix}   \|\mbc h_{kk}'\|^2\, \id_{2T} & |\mbc h_{kk}^{\prime\He}\mb h_{k\notk}'|\,  \mb V_k^\Tr \mb U_k \mb V_\notk \\
  |\mbc h_{kk}^{\prime\He}\mb h_{k\notk}'|\, \mb V_\notk^\Tr \mb U_k^\Tr \mb V_k &  \|\mb h_{k\notk}'\|^2\, \id_{2T} \end{bmatrix}
 \begin{bmatrix} \mb\Phi_k & ~ \\ ~ &  \mb\Phi_\notk \end{bmatrix}
 \right)
\\
\label{eq:MIMORnum5}
&=2^{-4T}
\det\left(
  \begin{bmatrix} \id_{2T} +2 \|\mbc h_{kk}'\|^2\, \mb\Phi_k& 2|\mbc h_{kk}^{\prime\He}\mb h_{k\notk}'|\,  \mb V_k^\Tr \mb U_k \mb V_\notk \mb\Phi_\notk \\
  2|\mbc h_{kk}^{\prime\He}\mb h_{k\notk}'|\, \mb V_\notk^\Tr \mb U_k^\Tr \mb V_k\mb\Phi_k & \id_{2T} +2  \|\mb h_{k\notk}'\|^2\, \mb\Phi_\notk \end{bmatrix}
 \right)
\\
\label{eq:MIMORnum6}
&=2^{-4T}
\det\left(\id_{2T} +2 \|\mb h_{k\notk}'\|^2\, \mb\Phi_\notk\right) \nonumber\\&\quad\quad\cdot
\det\left(
 \id_{2T} +2 \|\mbc h_{kk}'\|^2\, \mb\Phi_k
 - ( 2|\mbc h_{kk}^{\prime\He}\mb h_{k\notk}'|\,  \mb W_k \mb\Phi_\notk)
 (\id_{2T} +2 \|\mb h_{k\notk}'\|^2\, \mb\Phi_\notk)^{-1}
 (2|\mbc h_{kk}^{\prime\He}\mb h_{k\notk}'|\, \mb W_k^\Tr \mb\Phi_k )
 \right)
\\
\label{eq:MIMORnum7}
&=2^{-4T}
\det\left(\id_{2T} +2 \|\mb h_{k\notk}'\|^2\, \mb\Phi_\notk\right)
\det\left(
 \id_{2T} +  2 \|\mbc h_{kk}'\|^2\,  \mb W_k \mb D_k \mb W_k^\Tr  \mb\Phi_k
 \right)\ifCLASSOPTIONdraftcls.\else\fi
\end{align}
\end{subequations}
\ifCLASSOPTIONdraftcls\else\end{figure*}
\eqref{eq:MIMORnum} on the top of page~\pageref{eq:MIMORnum}. \fi
The reformulation in \eqref{eq:MIMORnum6} is obtained by taking the Schur complement \cite{HoJo13}
and makes use of the real unitary matrix $\mb W_k=\mb V_k^\Tr \mb U_k \mb V_\notk$.
In \eqref{eq:MIMORnum7}, we have defined the diagonal matrix
\begin{equation}
\mb D_k = \id_{2T}
 - \frac{|\mbc h_{kk}^{\prime\He}\mb h_{k\notk}'|^2}{\|\mbc h_{kk}'\|^2\,}   \mb\Phi_\notk
 \left(\frac{1}{2}\id_{2T} + \|\mb h_{k\notk}'\|^2\, \mb\Phi_\notk\right)^{-1}.
\end{equation}

Combining \eqref{eq:rkRT}, \eqref{eq:MIMORdenom}, and \eqref{eq:MIMORnum}, we have 
\begin{subequations}
\begin{align}
r_k &= \frac{1}{2T}\log_2 \det(
 \id_{2T} +  2 \|\mbc h_{kk}'\|^2\,  \mb W_k \mb D_k \mb W_k^\Tr  \mb\Phi_k)
 \\
 &\leq\frac{1}{2T}\log_2 \det(\id_{2T} + 2 \|\mbc h_{kk}'\|^2\, \mb D_k \mb {\tilde \Phi}_k) =: {\bar r}_k
\label{eq:enhanced:UB}
\end{align}
\end{subequations}
where the diagonal matrix $\mb {\tilde \Phi}_k$ is a reordered version of $\mb \Phi_k$ that is arranged in a way that the $i$th largest entry of $\mb {\tilde \Phi}_k$
is at the same position as the $i$th largest entry of $\mb D_k$.
The bound is due to the Hadamard inequality \cite[Sec.~7.8]{HoJo13} and due to the optimal ordering of $\mb {\tilde \Phi}_k$,
which can be shown in analogy to the optimality of channel pairing in the relay scenario in \cite{HaWi07}.\footnote{The main argument
can be summarized as follows.
Let $x_1\geq x_2\geq0$, $y_1\geq y_2\geq0$, and $a>0$. Then,
$\log(1+a{x_1}{y_1})+\log(1+a{x_2}{y_2})
\geq
\log(1+a{x_1}{y_2})+\log(1+a{x_2}{y_1})$
is equivalent to
$1+a{x_1}{y_1}+a{x_2}{y_2}+a^2{x_1x_2}{y_1y_2}
\geq
1+a{x_1}{y_2}+a{x_2}{y_1}+a^2{x_1x_2}{y_1y_2}$
$\Leftrightarrow$
$(x_1-x_2)({y_1}-{y_2})\geq 0$,
which is fulfilled.
Similar arguments were used in a real-valued SISO scenario in \cite{BeLiNaYa16}
and in a complex SISO scenario in \cite{HeUt19b}.}

Note that the $i$th diagonal element of $\mb D_k$ is nonincreasing in the $i$th diagonal element of $\mb\Phi_\notk$
and independent of the other elements of $\mb\Phi_\notk$.
Thus, the $i$th largest entry of $\mb D_k$ (which should be at the same position as the $i$th largest entry of $\mb\Phi_i$)
is at the position of the $i$th smallest entry of $\mb\Phi_\notk$.
Thus, if we can find $\mb V_1$ and $\mb V_2$ such that $\mb W_1=\mb W_2=\id_{2T}$,
the bound in \eqref{eq:enhanced:UB} is achievable for both users simultaneously by 
anti-aligned entries (one increasing, the other decreasing) in $\mb\Phi_\notk$ and $\mb\Phi_k$.

While finding such $\mb V_1$ and $\mb V_2$ is in general not possible, it is indeed possible for the case $\theta=0$.
In this case, we obtain from \eqref{eq:enhanced:Uk} that $\mb U_1=\mb U_2 =\id_{2T}$,
so that $\mb V_1=\mb V_2=\id_{2T}$ is adequate to achieve the upper bound.
As the upper bound does not depend on $\theta$ (since it does not depend on $\mb U_2$),
it is the same for the two systems \eqref{eq:modelTrans} and \eqref{eq:modelUB},
but it can be achieved with equality in the enhanced system \eqref{eq:modelUB}.
\end{proof}

\begin{proof}[Proof of Lemma~\ref{lem:proper_opt}]
In the proof of Lemma~\ref{lem:enhanced}, we have already shown using \eqref{eq:enhanced:UB} that the optimal rates ${\bar r}_k$ in the enhanced system \eqref{eq:modelUB} can be achieved with diagonal covariance matrices $\covR{\ubar{\crr x}_k}=\mb\Phi_k$
by choosing $\mb V_1=\mb V_2=\id_{2T}$.
Based on the diagonal entries
\begin{equation}
D_{k,t}^{(\ms)} = 1 - \frac{|\mbbar h_{kk}^\Tr\mb h_{k\notk}'|^2}{\|\mbbar h_{kk}\|^2\,
 \left(\frac{1}{2} + \|\mb h_{k\notk}'\|^2\, \Phi_{\notk,t}^{(\ms)}\right)}   \Phi_{\notk,t}^{(\ms)}
\end{equation}
of $\mb D_k^{(\ms)}$ in the $\ms$th strategy,
we have\footnote{We have substituted the diagonal entry $\tilde \Phi_{k,t}^{(\ms)}$ by $\Phi_{k,t}^{(\ms)}$ since the ordering will implicitly be optimized when optimizing the diagonal entries of $\mb \Phi_k^{(\ms)}$.}
\begin{subequations}
\label{eq:rate_2T}
\begin{align}
\label{eq:rate_2T1}
&{\bar r}_k(\mb\Phi_1^{(\ms)},\mb\Phi_2^{(\ms)}) \\
&=\frac{1}{2T} \sum_{t=1}^{2T} \log_2
 \left(1 + 2 \|\mbbar h_{kk}\|^2\, D_{k,t}^{(\ms)}  \Phi_{k,t}^{(\ms)} \right) 
\\
&=\frac{1}{2T} \sum_{t=1}^{2T} \log_2 
 \left(1 + \mbbar h_{kk}^\Tr\left( \mb h_{k\notk}' \Phi_{\notk,t}^{(\ms)} \mb h_{k\notk}^{\prime\Tr} + \frac{1}{2}\id_2\right)^{-1} \!\! \mbbar h_{kk}  \Phi_{k,t}^{(\ms)} \right)
\label{eq:rate_2T2}
\end{align}
\end{subequations}
where the last equality can be verified by applying the matrix inversion lemma (e.g., \cite[Sec.~0.7.4]{HoJo13}) to the inverse in \eqref{eq:rate_2T2}.
The power constraints \eqref{eq:primal:pow} can be expressed in terms of $\mb\Phi_1$ and $\mb\Phi_2$ as
\begin{equation}
\label{eq:pow_2T}
\sum_{\ms=1}^\Ms \tau_\ms \frac{1}{T} \tr{\mb \Phi_k^{(\ms)} }\leq P_k, ~~\allk.
\end{equation}

Using a similar argument as in \cite{NaNg19,HeUt19b}, we can use $L'=2TL$ time slots with $\tau_{\ms}'=\tau_{\lceil \frac{\ms}{2T} \rceil}/(2T)$ in \eqref{eq:primal},
and we can then set
\begin{align}
\Phi_{k,t}^{\prime(2T(\ms-1)+s)} = \Phi_{k,s}^{(\ms)}, \quad \forall t\in\{1,\dots,2T\}
\end{align}
for all $s\in\{1,\dots,2T\}$, $\ms\in\{1,\dots,L\}$, and $k\in\{1,2\}$.
This does not change the value on the left hand side of \eqref{eq:pow_2T},
and the left hand side of \eqref{eq:primal:rate} remains unchanged as well if \eqref{eq:rate_2T2} is used as rate expression.
Thus, there always exists an optimal solution with scaled identity matrices as covariance matrices $\covR{\ubar{\crr x}_k}$
of the composite real representation \eqref{eq:crrmat} of the extended symbol vectors \eqref{eq:symext}.
These scaled identities, however, imply that the symbol extensions have not been necessary and,
via \eqref{eq:crrcov}, show that the obtained strategy has vanishing pseudovariances in each time slot.
\end{proof}

\bibliographystyle{IEEEtran}
\bibliography{ConfIEEE,IEEEabrv,literature}

\end{document}